\documentclass[12pt]{article}
\usepackage{amsmath, amssymb, amsthm}
\usepackage[dvips]{graphicx}
\usepackage{natbib}
\usepackage{color}

\newtheorem{corollary}{Corollary}

\newtheorem{proposition}{Proposition}
\newtheorem{remark}{Remark}
\newtheorem{theorem}{Theorem}
\newtheorem{lemma}{Lemma}

\newcommand{\sog}{SO(3)}
\newcommand{\sogp}{SO(p)}
\newcommand{\og}{O(3)}
\newcommand{\ogp}{O(p)}
\newcommand{\ogr}{O(r)}
\newcommand{\stiefel}{V_2(\mathbb{R}^3)}
\newcommand{\stiefelrp}{V_r(\mathbb{R}^p)}
\newcommand{\stiefelpp}{V_{p-1}(\mathbb{R}^p)}
\newcommand{\sgv}{\rho}

\newcommand{\tr}{\operatorname{tr}}
\newcommand{\etr}{\operatorname{etr}}
\newcommand{\Rd}{\partial}
\newcommand{\diag}{\operatorname{diag}}
\newcommand{\sgn}{\operatorname{sgn}}

\newcommand{\mbR}{{\mathbb{R}}}

\usepackage{bm}
\bmdefine{\bx}{x}
\bmdefine{\btheta}{\theta}

\newcommand{\bX}{X}
\newcommand{\bTheta}{\Theta}

\newcommand{\hpm}{\hphantom{-}}

\def\commentA#1{ }

\title{Properties and applications of Fisher distribution on the rotation group}
\author{Tomonari Sei\thanks{Department of Mathematics, Keio University}, 
Hiroki Shibata\thanks{Department of Mathematical Informatics, Graduate School of Information Science and Technology, University of Tokyo} ,  
Akimichi Takemura\footnotemark[2]\ \thanks{JST, CREST} , \\
Katsuyoshi Ohara\thanks{Faculty of Mathematics and Physics, Kanazawa University} , 
Nobuki Takayama\thanks{Department of Mathematics, Kobe University}\ \footnotemark[3]} 
\date{June 2012}

\begin{document}
\maketitle

\begin{abstract}
We study properties of Fisher distribution (von Mises-Fisher distribution, 
matrix Langevin distribution) on the rotation group SO(3).
In particular we apply the holonomic gradient descent, introduced by Nakayama et al. (2011),
and a method of series expansion  
for evaluating the
normalizing constant of the distribution
and for computing the maximum likelihood estimate.
The rotation group 
can be identified with the Stiefel manifold
of two orthonormal vectors.
Therefore from the viewpoint of statistical modeling, it is of interest to compare 
Fisher distributions on these manifolds.
We illustrate the difference with an example of near-earth objects data.
\\
{\it Keywords: algebraic statistics; directional statistics;
holonomic gradient descent; maximum likelihood estimation; rotation group.}
\end{abstract}

\nocite{M2}
\nocite{asir}
\nocite{openxmmath}

\section{Introduction}
\label{sec:introduction}
In this paper we apply the holonomic gradient descent (HGD)
introduced in \cite{hgd} 
and a method of series expansion for evaluating the normalizing constant 
of Fisher distribution
on the rotation group and on Stiefel manifolds
and for obtaining the maximum likelihood estimate.
Fisher distribution is the most basic exponential 
family model for these manifolds.

The general theory of exponential families is  well established 
(e.g.\ \cite{barndorff-nielsen}).  In nice ``textbook'' cases, the normalizing
constant of the exponential family (i.e.\ its cumulant generating function)
can be explicitly evaluated and
then the calculation of maximum likelihood estimate is also simple.  However in 
general, the integral  defining the normalizing constant of an exponential family 
can not be explicitly evaluated. 
Various  techniques, such as infinite series expansion, numerical integration, 
Markov chain Monte Carlo methods,  iterative proportional scaling, 
have been applied for these cases.

Recently, we introduced a very novel approach, the holonomic gradient descent,
for evaluation of the normalizing constant and solving the likelihood equation (\cite{hgd}).
Our approach provides a systematic methodology for these tasks.
Note that the normalizing constant is a definite integral over the sample space, where the integrand contains the parameter of the family of distributions.  
The likelihood equation involves
differentiation of the normalizing constant with respect to the parameter.
In the holonomic gradient descent, 
the theory of $D$-modules is used to derive a set of partial differential equations
satisfied by the normalizing constant. 

We illustrate the HGD method for a simple example.
Let $C(\kappa)$, $\kappa\geq 0$, be the normalizing constant of the von Mises-Fisher distribution on $S^1=\{(x_1,x_2)\mid x_1^2+x_2^2=1\}$
with the density function
\begin{align*}
 p(x|\kappa)
 = C(\kappa)^{-1}\exp\left(\kappa x_1\right),
 \quad x=(x_1,x_2)\in S^1,
\end{align*}
with respect to the uniform distribution on $S^1$.
We assumed that the mean direction is known to be $(1,0)$ just for simplicity.
For the sample mean $\bar{x}=(\bar{x}_1,\bar{x}_2)$ of a given data set on $S^1$ of size $n$,
the log-likelihood function is given by $n(\kappa\bar{x}_1-\log C(\kappa))$.
We calculate the maximum likelihood estimate by gradient descent methods,
that is, update the parameter $\kappa$ according to
$\kappa\leftarrow \kappa + A\partial\{\kappa\bar{x}_1-\log C(\kappa)\}$
with some (fixed or adaptive) number $A$,
where $\partial=d/d\kappa$ denotes the derivative with respect to $\kappa$.
This needs evaluation of $C(\kappa)$ and its derivatives for a number of points $\kappa$.
The HGD is the gradient descent method together with the following algorithm
of computing $C(\kappa)$ and its derivatives.
First we note that the function $C(\kappa)$ is the modified Bessel function of the first kind
and satisfies a differential equation
\begin{align}
 \label{eq:vMF-S1}
 \partial^2 C + \frac{1}{\kappa}\partial C - C = 0
\end{align}
(see e.g.\ \cite{MardiaJupp2000}).
Assume that we know the values of $C$ and $\partial C$ at a given point $\kappa^{(0)}>0$.
Then values of $C$ and $\partial C$ at any other point $\kappa^{(1)}>0$
is obtained if one solves the differential equation
\begin{align}
 \label{eq:vMF-S1-simul}
 \partial
 \begin{pmatrix}
  C\\ \partial C
 \end{pmatrix}
 = 
 \begin{pmatrix}
  0& 1\\
  1& -\kappa^{-1}
 \end{pmatrix}
 \begin{pmatrix}
  C\\ \partial C
 \end{pmatrix}
\end{align}
numerically from $\kappa=\kappa^{(0)}$ to $\kappa=\kappa^{(1)}$.
A key point is that the coefficient matrix in (\ref{eq:vMF-S1-simul}) is easily evaluated.
Now the HGD algorithm is given as follows:
\begin{itemize}
 \item[1.] Compute $C(\kappa^{(0)})$ and $\partial C(\kappa^{(0)})$ at a point $\kappa^{(0)}$
	   by some method such as numerical integration or series expansion.
 \item[2.] For $t=1,2,\ldots$, repeat the following procedure
	   until $\kappa^{(t)}$ converges.
	   \begin{itemize}
	    \item[2-a.] Determine $\kappa^{(t)}$ from $\kappa^{(t-1)}$
			according to the gradient method.
	    \item[2-b.] Solve the equation (\ref{eq:vMF-S1-simul}) numerically
			from $\kappa^{(t-1)}$ to $\kappa^{(t)}$.
	   \end{itemize}
\end{itemize}
The method needs direct computation of $C(\kappa)$ and $\partial C(\kappa)$
only once, at $\kappa=\kappa^{(0)}$.

The method is also available for multi-dimensional parameters.
Assume that a function $C(\theta)$ of $\theta\in\mathbb{R}^d$, typically the normalizing constant
of a parametric family,
satisfies the following partial differential equations:
\begin{align}
 \label{eq:Pfaffian-system}
 \frac{\partial}{\partial\theta_i}G
 &= P_i(\theta)G,
 \quad i=1,\ldots,d,
\end{align}
where $G=G(\theta)$ is a finite-dimensional vector consisting of
some partial derivatives of $C(\theta)$ and $P_i(\theta)$ is a square matrix of rational functions of $\theta$ for each $i$.
An example is the equation (\ref{eq:vMF-S1-simul}), where $\theta=\kappa$.
Note that (\ref{eq:Pfaffian-system}) is essentially an {\em ordinary} differential equation
because each equation contains only one $\partial_i$.
The equation (\ref{eq:Pfaffian-system}) is called {\em a Pfaffian system} of $C(\theta)$.
Let $\theta^{(0)}$ and $\theta^{(1)}$ be two points in $\mathbb{R}^d$
and assume that a numerical value of $G(\theta^{(0)})$ is known.
Then we can obtain a numerical value of $G(\theta^{(1)})$ as follows.
By recursive argument, it is sufficient to consider the case that $\theta^{(0)}$ and $\theta^{(1)}$
have the same components all but the $i$-th component for some $i$.
Then a numerical value of $G(\theta^{(1)})$ is obtained
by solving the equation (\ref{eq:Pfaffian-system}) for $i$
with the initial condition of $G(\theta^{(0)})$.
Now the HGD algorithm is the same as the one-dimensional case.

In this paper we apply the holonomic gradient descent 
and a method of series expansion
to Fisher distribution on the rotation group $\sog$ and on
the Stiefel manifold $\stiefel$ of two orthonormal vectors.  The Fisher
distribution on Stiefel manifolds and orthogonal groups 
has been studied by number of authors.
However only a few papers (\cite{prentice-1986}, \cite{wood-1993})
study the Fisher distribution on the special orthogonal group $\sogp$.

The holonomic gradient descent needs the initial value
for the differential equation as illustrated above.
To evaluate this value, we develop an explicit formula
of the infinite series expansion for $\sog$.
An alternative method is a one-dimensional integration formula
proposed by \cite{wood-1993}.
Furthermore, as a referee pointed out, 
the Fisher distribution on $\sog$ is identified with the Bingham distribution on
the real projective space $\mathbb{R}P^3$ (\cite{prentice-1986}, \cite{wood-1993}, \cite{MardiaJupp2000}).
The normalizing constants of the Bingham distributions
are hypergeometric functions of a matrix argument and
saddlepoint approximations to these normalizing constants were given by \cite{KumeWood2005}.

The organization of the paper is as follows.  For the rest of this section
we set up notation and summarize preliminary facts on special orthogonal groups and 
Stiefel manifolds.
In Section \ref{sec:fisher-distribution} we derive some properties
of Fisher distribution on special orthogonal groups and Stiefel manifolds.
In Section \ref{sec:two-methods}  we derive the set of partial differential
equations satisfied by the normalizing constant (Section \ref{subsec:holonomic}).  
We also give an infinite series expansion
for the normalizing constant (Section \ref{subsec:series-expansion}).
In Section \ref{sec:data-analysis} we apply the results of previous
sections to the data on orbits of near-earth objects.


\subsection{Notation and preliminary facts}
\label{subsec:notation}
Here we set up notation of this paper
and summarize some preliminary facts.
Although we are primarily interested in $3\times 3$
matrices 
for practical and computational reasons, 
we set up our notation for general dimension.
%
Let
\[
\stiefelrp= \{ A \in \mathbb{R}^{p\times r} \mid A^\top A=I_r \}
\quad (0 < r\le p)
\]
denote  the 
Stiefel manifold of $p\times r$ real matrices with orthonormal columns,
where 
$\mathbb{R}^{p\times r}$ denotes the set of
$p\times r$ real  matrices and $A^\top$ denote the transpose of $A$.
In particular for $r=p$,
\[
\stiefelrp=\ogp
\]
is the set of $p\times p$ orthogonal
matrices. 
\[
\sogp = \{ X \in \ogp \mid \det X = 1 \}
\]
denotes the special orthogonal group.


Let ${\rm Vol}$ denote the invariant measure (volume element) on 
$\stiefelrp$ and
let 
\[
\mu(\cdot)=\frac{1}{{\rm Vol}(\stiefelrp)} {\rm Vol}(\cdot)
\]
denote the invariant probability measure on $\stiefelrp$.
We also call $\mu$ the uniform distribution.
Similarly for $\sogp$,  by 
$\mu(\cdot)={\rm Vol}(\cdot)/{\rm Vol}(\sogp)$ we denote the uniform distribution with 
\[
{\rm Vol}(\sogp)=\frac{1}{2}{\rm Vol}(\ogp).
\]
If $p\geq 3$ and $X=(X_{ij})$ is uniformly distributed on $\sogp$,
then we can show that the first two moments are
\begin{align}
 \label{eq:sogp-moment}
 E[X_{ij}]=0
 \ \ \mbox{and}\ \ 
 E[X_{ij}X_{kl}] = \frac{1}{p}1_{\{(i,j)=(k,l)\}}.
\end{align}
We prove $E[X_{11}]=E[X_{11}X_{21}]=0$ and $E[X_{11}^2]=1/p$.
The other cases are similarly proved.
Let $Q=\diag(-1,1,\cdots,1,-1)\in\sogp$.
Then $(QX)_{11}=-X_{11}$ and $(QX)_{21}=X_{21}$.
Since $X$ and $QX$ have the same distribution,
we have $E[X_{11}]=0$ and $E[X_{11}X_{21}]=0$.
By a similar reason, $X_{11}$ and $X_{i1}$ for each $i$
have the same marginal distribution.
Therefore $E[X_{11}^2]=E[X_{i1}^2]$.
Since $\sum_{i=1}^pX_{i1}^2=1$, we have $E[X_{11}^2]=1/p$.

For a $p\times r$ matrix $\Theta \in \mbR^{p\times r}$, $r\le p$, let
\[
\Theta = Q D R,  \quad Q\in \stiefelrp, \ D=\diag(\sgv_1,\ldots,\sgv_r), \ R\in \ogr
\]
denote its singular value decomposition (SVD),
where $\sgv_1\geq\dots\geq\sgv_p\geq 0$.
Now let $\Theta \in {\mbR}^{p\times p}$ be a square matrix and $\Theta=QDR$ be the SVD.
Let $\epsilon=\det(QR)\in\{-1,1\}$.
Then we can write
\begin{equation}
\label{eq:svd-sop}
\Theta = \tilde{Q} \tilde{D} \tilde{R},  \quad \tilde{Q}, \tilde{R}\in \sogp, \ \tilde{D}=\diag(\epsilon \sgv_1, \sgv_2, \dots, \sgv_p),
\end{equation}
where $\tilde{Q}=Q\diag(\det Q,1,\dots,1)$ and $\tilde{R}=\diag(\det R,1,\dots,1)R$.
If $\Theta$ is non-singular, we have $\epsilon=\sgn\det\Theta$.
We call the decomposition (\ref{eq:svd-sop}) the {\em sign-preserving SVD} of $\Theta$ with respect to $\sogp$.
We also call 
$\phi_1 = \epsilon \sgv_1$, $\phi_i = \sgv_i$, $i\ge 2$, the {\em sign-preserving
singular values} of $\Theta$.
The decomposition is also used in \cite{prentice-1986} and \cite{wood-1993}.

\section{Fisher distributions on $\stiefelrp$ and $\sogp$}
\label{sec:fisher-distribution}
In this section we consider Fisher distribution on 
$\stiefelrp$ and $\sogp$.  In particular 
we clarify the difference
between Fisher distributions on $\stiefelpp$ and $\sogp$.
Basic facts on Fisher distribution on $\stiefelrp$ are
summarized in Chapter 13 of \cite{MardiaJupp2000}.

Let ${\cal X}$ denote either $\stiefelrp$ or $\sogp$.
The density of the Fisher distribution on ${\cal X}$
with respect to the uniform distribution $\mu$ is given by
\[
 f(\bX;\bTheta) = \frac{1}{c(\bTheta)}\etr(\bTheta^\top \bX), \quad X\in {\cal X},
\]
where  $\bTheta = (\theta_{ij}) \in \mathbb{R}^{p\times r}$
is  the parameter matrix, $\etr(\cdot)=\exp(\tr(\cdot))$,  and
\begin{equation}
\label{eq:norm-const}
c(\bTheta) = \int_{{\cal X}} \etr(\bTheta^\top \bX) \mu(d\bX)
\end{equation}
is the normalizing constant.  For $\stiefelrp$ it is well known 
(e.g.\ \cite{KhatriandMardia1977}, \cite{muirhead82}, \cite{chikuse03}) that
$c(\Theta)$ is a hypergeometric function
$c(\Theta)={}_0F_1(p/2,Y)$ with a matrix argument, 
where $Y=\Theta^\top \Theta/4$.  
However properties of $c(\Theta)$ for the special orthogonal
group ${\cal X}=\sogp$ have not been studied in detail.
For the case of $\sog$, 
following the approach in \cite{prentice-1986}, \cite{wood-1993} 
used the correspondence between the Fisher distribution on $\sog$ and the Bingham distribution on
the unit sphere $S^3$ or the real projective space $\mathbb{R}P^3$ and showed that
$c(\Theta)$ can be written as a one-dimensional integral
involving the modified Bessel function of degree zero.
The normalizing constants of the Bingham distributions
are hypergeometric functions of a matrix argument and
saddlepoint approximations to these normalizing constants were given by \cite{KumeWood2005}.
In Section \ref{sec:two-methods}
we derive differential equations and an infinite series expansion of $c(\Theta)$ 
for $\sog$.

Let $\bx_1,\dots, \bx_p$ be the columns of $X\in \sogp$.  Since
$\bx_p$ is uniquely determined from $\bx_1,\dots, \bx_{p-1}$, we can 
identify $\sogp$ with $\stiefelpp$ by
\begin{equation}
\label{eq:so-stiefel-identificatin}
(\bx_1,\dots, \bx_p) \in \sogp  \ \leftrightarrow \ 
(\bx_1,\dots, \bx_{p-1})\in \stiefelpp 
\end{equation}
This leads to the question of differences between Fisher distributions
on $\sogp$ and those on $\stiefelpp$.  
Let $\Theta = (\btheta_1, \dots, \btheta_p) \in \mbR^{p\times p}$  be
a parameter matrix for Fisher distribution on $\sogp$. By
setting $\btheta_p=0$, we clearly obtain a Fisher distribution on $\stiefelpp$.
Hence the family of Fisher distributions on $\stiefelpp$ is a submodel
of the family of Fisher distributions on $\sogp$.
It can be easily seen that for $p=2$, $\btheta_2$ is redundant and
these two families are the same. 
However for $p\ge 3$, the family of Fisher distributions on
$\stiefelpp$ is a strict submodel of that on $\sogp$.
We state this as a lemma.

\begin{lemma}
For $p\ge 3$, the family of Fisher distributions on $\stiefelpp$
is a strict submodel of that on $\sogp$.
\end{lemma}

\begin{proof}
In general, let $K$ be a positive integer and consider a $K$-dimensional exponential family 
\begin{align*}
 p(x|\theta) = \frac{1}{C(\theta)}\exp\left(\theta^{\top}x\right)\quad
 (x\in S),\quad
 C(\theta) = \int_S \exp(\theta^{\top}x)\nu(dx),
\end{align*}
where $\theta$ is a $K$-dimensional vector, $S$ is a smooth submanifold  of $\mbR^K$ and
$\nu$ is a measure on $S$. Assume that
$C(\theta)$ exists in some open neighborhood of the origin $\theta=0$.
We call the parameter $\theta$ is estimable if $p(x|\theta)\neq p(x|\tilde\theta)$ as functions of $x$
whenever $\theta\neq\tilde\theta$.
For the exponential family, $\theta$ is estimable if and only if 
\begin{align*}
 {\rm Affine}({\rm support}(\nu))
 \ =\ \mbR^K
\end{align*}
(e.g.\ Corollary 8.1 of \cite{barndorff-nielsen}), where 
${\rm support}(\nu)$ is the support of $\nu$ and 
${\rm Affine}(U)$, $U\subset \mbR^K$, denotes the affine hull of $U$.

We now show that if $p\geq 3$ then
$
 {\rm Affine}(\sogp) = \mbR^{p\times p}
$
and therefore 
$\Theta$ is estimable, which is sufficient to prove the lemma.
Let $L={\rm Affine}(\sogp)$. 
We first see that the zero matrix $0$ belongs to $L$.
Let $\epsilon_i\in\{-1,1\}$ for $1\le i\le p$.
Then the average of $2^{p-1}$ matrices
$\diag(\epsilon_1,\ldots,\epsilon_{p-1},\prod_{i=1}^{p-1}\epsilon_i)$
in $\sogp$ is zero. Hence $0\in L$.
We now show that $e_ie_j^{\top}$ belongs to $L$ ($\forall i, j$), where
$e_i=(0,\ldots,1,\ldots,0)^{\top}$ is the
standard basis vector with $1$ as the $i$-th element.  Then
together with $0\in L$ it follows that $L=\mbR^{p\times p}$.
Take matrices $P_i\in\sogp$ ($i=1,\ldots,p$) such that $P_ie_i=e_1$.
For example, let $P_1=I_p$ and $P_i=e_1e_i^{\top}-e_ie_1^{\top}+\sum_{j\neq 
 1,i}e_je_j^{\top}$ for $i\neq 1$.
Then $e_ie_j^{\top}\in L$ if and only if
$e_1e_1^{\top}=P_ie_ie_j^{\top}P_j^{\top}\in L$.
Now it suffices to show that $e_1e_1^{\top}$ belongs to $L$.
Take the average of $2^{p-2}$ matrices $\diag(1,\epsilon_2,\ldots,\epsilon_{p-1},\prod_{i=2}^{p-1}\epsilon_i)$.
Then $e_1e_1^{\top}\in L$.
\end{proof}

For ${\cal X}=\stiefelrp$
the maximum likelihood estimate (MLE) of the Fisher distribution
is obtained by the following procedure (\cite{KhatriandMardia1977}).
Let $X^{(1)},\ldots,X^{(N)}$ be a data set
on $\stiefelrp$.
Let $\bar X=N^{-1}\sum_{t=1}^NX^{(t)}$ be the sample mean matrix
and let $\bar X=Q\diag(g_1,\ldots,g_r)R$ be the SVD of $\bar{X}$,
where $Q\in\stiefelrp$, $R\in\ogr$ and $g_1\ge \cdots\ge g_r\ge 0$.
Then the maximum likelihood estimate $\hat \Theta$ is given by
$\hat \Theta=Q\diag(\hat\phi_1,\ldots,\hat\phi_r)R$, where
$\hat\phi_i$ is the solution of
\[
 \frac{\int_{\stiefelrp}x_{ii}\exp(\sum_{k=1}^r\hat\phi_k x_{kk})\mu(dX)}
 {\int_{\stiefelrp}\exp(\sum_{k=1}^r \hat\phi_k x_{kk})\mu(dX)}
 = g_i, \quad i=1,\dots,r.
\]

This procedure is also valid for $\sogp$ if we use the sign-preserving SVD in 
(\ref{eq:svd-sop}).  We give the fact as a lemma since it is not
explicitly proved in the literature.
Remark that for $\sogp$ the normalizing constant $c(\Theta)$
in (\ref{eq:norm-const}) is invariant under
a transformation $\Theta\mapsto Q\Theta R$ for any $Q,R\in\sogp$.

\begin{lemma} \label{lemma:mle-diag}
Let $X^{(1)},\ldots,X^{(N)}$ be a data set on $\sogp$.
Let $\bar X=N^{-1}\sum_{t=1}^NX^{(t)}$ be the sample mean matrix and
$\bar X=Q\diag(g_1,\ldots,g_p)R$ be the sign-preserving SVD of $\bar X$,
where $Q,R\in\sogp$ and $|g_1|\ge g_2 \cdots\ge g_p\ge 0$.
Then the maximum likelihood estimate of the Fisher distribution on $\sogp$
is $\hat \Theta=Q\diag(\hat\phi_1,\ldots,\hat\phi_p)R$, where
$\hat\phi_i$ is the maximizer of the function
\begin{equation}
 \label{eq:objective}
 (\phi_k)_{k=1}^p\mapsto\sum_{k=1}^p\phi_k g_k - \log c(\diag(\phi_1,\ldots,\phi_p)),
\end{equation}
or equivalently, the solution of
\begin{equation}
\label{eq:likeq}
 \frac{\int_{\sogp}x_{ii}\exp(\sum_{k=1}^p \hat\phi_k x_{kk})\mu(dX)}
 {\int_{\sogp}\exp(\sum_{k=1}^p\hat\phi_k x_{kk})\mu(dX)}
 = g_i, \quad i=1,\dots,p.
\end{equation}
\end{lemma}

\begin{proof}
We change the parameter variable from $\Theta$
to $\Phi=(\phi_{ij})_{i,j=1}^p=Q^\top\Theta R^\top$.
Then the ($1/N$ times) log-likelihood function is written
as
\begin{equation}
\label{eq:likeq-sogp}
\tr(\Theta^{\top}\bar{X}) - \log c(\Theta)
=\tr(\Phi^{\top}G) - \log c(\Phi),
\end{equation}
where $G=\diag(g_1,\ldots,g_p)$.
Since (\ref{eq:likeq-sogp}) is strictly convex in $\Phi$,
the unique maximizer makes its first-order derivatives zero.
Note that the first term on the right hand side of (\ref{eq:likeq-sogp})
does not depend on the off-diagonal elements of 
$\Phi$.  Therefore the condition for 
maximization of (\ref{eq:likeq-sogp})
with respect to an off-diagonal element is written as
\begin{align}
 0 = \frac{\partial}{\partial \phi_{ij}}\log c(\Phi),
 \quad (i\neq j)
 \label{eq:mle-off-diagonal}.
\end{align}
We now fix $i\neq j$ and evaluate $(\partial/\partial \phi_{ij})\log c(\Phi)$
at $(\phi_{i'j'})_{i'\neq j'}=0$.  Then we have
\begin{align*}
 \left.
 \frac{\partial}{\partial \phi_{ij}}\log 
 c(\Phi)
 \right|_{\phi_{i'j'}=0\ \forall i'\neq j'}
 \ =\ \frac{\int_{\sogp}x_{ij}\exp(\sum_{k=1}^p \phi_{kk}x_{kk})\mu(dX)}
 {\int_{\sogp}\exp(\sum_{k=1}^p\phi_{kk}x_{kk})\mu(dX)}.
\end{align*}
However
\begin{align*}
 \int_{\sogp}x_{ij}\exp(\sum_{k=1}^p\phi_{kk}x_{kk})\mu(dX)
 \ =\ \int_{\sogp}(-x_{ij})\exp(\sum_{k=1}^p \phi_{kk}x_{kk})\mu(dX)
 \ =\ 0
\end{align*}
because the uniform distribution $\mu$ on $\sogp$ is invariant
with respect to multiplication of the $i$-th row and the $i$-th (not $j$-th)
column of $X$ by $-1$, which transforms $x_{ij}$ into $-x_{ij}$ and $x_{kk}$ into $x_{kk}$ for every $k$.
Therefore any diagonal matrix $\Phi$ satisfies
(\ref{eq:mle-off-diagonal}).
The log-likelihood function of the diagonal matrix is (\ref{eq:objective}) and
the maximizer satisfies (\ref{eq:likeq}).
\end{proof}

When $\det \bar X < 0$, 
it is not correct to use the ordinary singular values of $\bar X$ on the
right-hand side of  (\ref{eq:likeq}).

\begin{remark} \label{remark:sogp-asymp}
The determinant of the sample mean matrix $\bar X$
is not necessarily positive even if all $X^{(t)}$, $t=1,\dots,N$, are in $\sogp$.
Indeed for the case of uniform distribution on $\sogp$ we prove 
\[
P(\det\bar X<0) \rightarrow \frac{1}{2}, \qquad (N\rightarrow\infty),
\]
as long as $p\geq 3$. By the central limit theorem
$\sqrt{N}(\bar X-E(X))$ converges to a Gaussian random matrix $Z$
with the same covariances as $X$.
From (\ref{eq:sogp-moment}), we know that $E(X)=0$ and
the covariances of $X$ are diagonal when $p\geq 3$.
Then $Z$ and any sign change of a column of $Z$ have the same probability distribution
and therefore the probability of $\det(Z)<0$ is $1/2$.
Hence the probability of $\det(\bar X)<0$ converges to $1/2$.
\end{remark}

\begin{remark}
 Even if $\det\bar{X}>0$,
 the determinant of the estimated parameter $\hat{\Theta}$ may be negative.
 Indeed, let the sign-preserving singular values of $\bar{X}$ and $\hat{\Theta}$
 be $g=(g_1,g_2,g_3)$ and
 $\hat{\phi}=(\hat{\phi}_1,\hat{\phi}_2,\hat{\phi}_3)$, respectively.
 We prove that $g_1g_2g_3$ and $\hat{\phi}_1\hat{\phi}_2\hat{\phi}_3$
 can have the opposite signs.
 To see this, we first consider the case $\hat{\phi}_1=0$,
 $\hat{\phi}_2>0$ and $\hat{\phi}_3>0$.
 Then, by using the Taylor expansion formulas (\ref{eq:ctheta-expansion}) 
 and (\ref{eq:eklm}) developed in Subsection~\ref{subsec:series-expansion},
 we deduce that $g_1$, $g_2$ and $g_3$ are strictly positive.
 By continuity, there exist some $\hat{\phi}_1<0$, $\hat{\phi}_2>0$ and $\hat{\phi}_3>0$
 while all $g_i$'s are positive.
\end{remark}

\section{Computation of the normalizing constant and its derivatives}
\label{sec:two-methods}
For computing the maximum likelihood estimate of Fisher distribution we need
numerical evaluation of the normalizing constant  $c(\bTheta)$ of 
(\ref{eq:norm-const}) and its derivatives.  In this section we study two
methods for this purpose.  The first method is the holonomic gradient
descent.
In the second method, we use 
series expansion of $\etr(\bTheta^\top \bX)$.
The second method is also used to
compute the initial value of HGD.

\subsection{The holonomic gradient descent for Stiefel manifolds and special orthogonal group}
\label{subsec:holonomic}

Let us briefly describe the holonomic gradient descent.
As to details, we refer to \cite{hgd}.
An algebraic computation is the first step;
we construct linear ODE's (ordinary differential equations) satisfied by $c(\Theta)$
with respect to each $\theta_{ij}$ by Gr\"obner bases of a set of partial differential
equations satisfied by $c$.
Variables other than $\theta_{ij}$ appear as parameters in the ODE.
The rank of ODE's is called the holonomic rank.
The ODE's give a dynamical system for the function
$c(\Theta)\etr(-\Theta^{\top}\bar X)$,
the reciprocal of the likelihood.
The gradient of the function can also be expressed in terms of derivatives
of the reciprocal standing for the standard monomials. 
The second step is a numerical procedure;
a point in the dynamical system moves
toward the maximum likelihood estimate
along the gradient direction,
simultaneously updating the values of $c(\Theta)$ and its derivatives.

For the holonomic gradient descent, we 
study differential operators $A$ annihilating $c(\Theta)$,
that is, $A\cdot c(\Theta)=0$.
Denote 
the differential operator
$\Rd/\Rd \theta_{ij}$ by $\Rd_{ij}$.  We first study the special orthogonal group
and then study the Stiefel manifold.

\subsubsection{The case of special orthogonal group}
Let $\Theta\in\mathbb{R}^{p\times p}$.
We consider the following three types of differential operators:
\begin{align*}
 A_{ij}^{(1)} &= \sum_{k=1}^p \Rd_{ik}\Rd_{jk}-\delta_{ij},
 \quad \tilde A_{ij}^{(1)} = \sum_{k=1}^p \Rd_{ki}\Rd_{kj}-\delta_{ij}\quad (i\leq j),
 \\
 A^{(2)} &= \det(\Rd_{ij}) - 1,
 \\
 A_{ij}^{(3)} &= \sum_{k=1}^p 
 \left(-\theta_{jk}\Rd_{ik}+\theta_{ik}\Rd_{jk}\right),
 \quad \tilde A_{ij}^{(3)} = \sum_{k=1}^p
 \left(-\theta_{kj}\Rd_{ki}+\theta_{ki}\Rd_{kj}\right)\quad (i<j),
\end{align*}
where $\delta_{ij}$ is the Kronecker's delta. 
The following lemma is an analog of Theorem 2 of \cite{hgd}.

\begin{lemma} \label{lemma:sogp-annihilate}
 The above differential operators annihilate $c(\Theta)$ of $\sogp$.
\end{lemma}

\begin{proof}
 We first prove that the operators $A_{ij}^{(1)}$, $\tilde A_{ij}^{(1)}$ and $A^{(2)}$
 annihilate $\etr(\Theta^{\top}X)$ for any $X\in\sogp$.
 Then they also annihilate $c(\Theta)$ because
 $A\cdot c(\Theta)=\int_{\sogp}A\cdot \etr(\Theta^{\top}X)\mu(dX)$
 for any operator $A$.
 Since $\Rd_{ij}\cdot\etr(\Theta^{\top}X)=x_{ij}\etr(\Theta^{\top}X)$
 and $XX^{\top}=I$, we have
 \begin{align*}
  A_{ij}^{(1)}\cdot\etr(\Theta^{\top}X)
  &= \left(\sum_{k=1}^p x_{ik} x_{jk}-\delta_{ij}\right)\etr(\Theta^{\top}X)
  = 0.
 \end{align*}
 Similarly, we obtain $\tilde A_{ij}^{(1)}\cdot\etr(\Theta^{\top}X)=0$ from $X^{\top}X=I$
 and $A^{(2)}\cdot\etr(\Theta^{\top}X)=0$ from $\det(X)=1$.

 Next consider $A_{ij}^{(3)}$ and $\tilde A_{ij}^{(3)}$.
 We note $c(\Theta)=c(Q\Theta)=c(\Theta Q)$ for any $Q\in\sogp$.
 For any fixed $i<j$, define a rotation matrix $Q=Q(\epsilon)$ by
 \begin{align*}
  Q
  &= (\cos\epsilon)(E_{ii}+E_{jj}) + (\sin\epsilon)(-E_{ij}+E_{ji})
  + \sum_{k\neq i,j}E_{kk},
 \end{align*}
 where $E_{kl}$ is the matrix whose $(i,j)$-th component
 is 1 if $k=i$ and $l=j$ and 0 otherwise.
 Then
 \begin{align*}
  0
  &= c(Q\Theta) - c(\Theta)
  \\
  &= c\left(\Theta - \epsilon\sum_k\theta_{jk}E_{ik}+\epsilon\sum_k\theta_{ik}E_{jk}+\mathrm{o}(\epsilon)\right) - c(\Theta)
  \\
  &= \epsilon\sum_{k=1}^p(-\theta_{jk}\Rd_{ik}+\theta_{ik}\Rd_{jk})\cdot c(\Theta)+\mathrm{o}(\epsilon),
 \end{align*}
 as $\epsilon\to 0$.
 Hence we have $A_{ij}^{(3)}\cdot c(\Theta)=0$.
 Similarly we obtain $\tilde A_{ij}^{(3)}\cdot c(\Theta)=0$ from $c(\Theta Q)=c(\Theta)$.
\end{proof}

Let $D$ be the ring of differential operators
with polynomial coefficients in $\theta_{ij}$
and let $I$ denote the ideal generated by the above
differential operators $A_{ij}^{(1)},\ldots,\tilde{A}_{ij}^{(3)}$ in $D$.
Also let
$I_{\rm diag}$ denote $I$ restricted to
diagonal matrices
$\Theta=\mathrm{diag}(\theta_{11},\ldots,\theta_{pp})$.
$I\cdot f(\Theta)=0$ implies $I_{\rm diag}\cdot
f(\mathrm{diag}(\Theta))=0$.
We denote by $R_p$ the ring of differential operators with rational
function coefficients in $\theta_{ij}$, $1 \leq i,j \leq p$.

The following propositions are essential for the holonomic gradient descent.
We refer to \cite{hgd} for the definition of holonomic ideals in $D$
and zero-dimensional ideals in $R_p$.
Once zero-dimensionality of $R_pI$ is proved and a Gr\"obner basis is constructed,
we can find ODE's and apply the holonomic gradient descent 
for the maximum likelihood estimate.

\begin{proposition}
If $p=2$, then the ideal $I$ is holonomic.
In particular, the ideal $R_2 I$ is zero-dimensional.
The holonomic rank is equal to $2$.
\end{proposition}

The proposition is proved by Macaulay2 (\cite{M2})
and the yang package on Risa/Asir (\cite{asir})
by utilizing Gr\"obner basis computations in rings of differential operators.
Also the set of generators of $I$ is obtained by nk\_restriction function of asir 
from the integral representation of $c(\Theta)$
 as
\begin{align*}
 & g_1 = -\Rd_{12}-\Rd_{21},
 \quad g_2 = -\Rd_{11}+\Rd_{22},
 \quad g_3 = \Rd_{21}^2+\Rd_{22}^2-1,
 \\
 & g_4=(\theta_{22}+\theta_{11})\Rd_{21}+(-\theta_{21}+\theta_{12})\Rd_{22},
 \\
 & g_5=(\theta_{21}-\theta_{12})\Rd_{22}\Rd_{21}+(\theta_{22}+\theta_{11})\Rd_{22}^2+\Rd_{22}-\theta_{22}-\theta_{11},
 \\
 & g_6=(-\theta_{21}+\theta_{12})\Rd_{21}+(\theta_{21}^2-2\theta_{12}\theta_{21}+\theta_{22}^2+2\theta_{11}\theta_{22}+\theta_{11}^2+\theta_{12}^2)
 \Rd_{22}^2
 \\
 &\quad +(\theta_{22}+\theta_{11})\Rd_{22}-\theta_{22}^2-2\theta_{11}\theta_{22}-\theta_{11}^2.
\end{align*}
Furthermore the set of generators of $I_{\rm diag}$ is given as
\begin{align*}
 h_1 = (-\theta_{22}-\theta_{11})\Rd_{11}^2-\Rd_{11}+\theta_{22}+\theta_{11},
 \quad h_2 = -\Rd_{11}+\Rd_{22}.
\end{align*}

\begin{proposition} \label{prop:2}
If $p=3$, then the ideal $R_3 I$ is zero-dimensional.
The holonomic rank is less than or equal to $4$.
$R_3/(R_3I)$ is spanned by $1, \partial_{31}, \partial_{32}, \partial_{33}$
as a vector space over the field of rational functions. 
\end{proposition}

The proposition is proved by a large scale computation on Risa/Asir with Gr\"obner bases.
The algorithm for it is explained in, e.g., \cite{hgd}.
Programs and obtained data are at the website OpenXM/Math (\cite{openxmmath}).
We conjecture that $I$ is holonomic and consequently
$R_p I$ is zero-dimensional for any $p$ in the case of $SO(p)$.

\subsubsection{The case of Stiefel manifold}

Let $\Theta\in\mathbb{R}^{p\times r}$ ($r\leq p$).
Consider the following differential operators:
\begin{align*}
 A_{ij}^{(1)} &= \sum_{k=1}^p \Rd_{ki}\Rd_{kj}-\delta_{ij}
 \quad (1\leq i\leq j\leq r),
 \\
 A_{ij}^{(2)} &= \sum_{k=1}^r\left(-\theta_{jk}\Rd_{ik}+\theta_{ik}\Rd_{jk}\right)
 \quad (1\leq i<j\leq p),
 \\
 \tilde A_{ij}^{(2)} &= \sum_{k=1}^p \left(-\theta_{kj}\Rd_{ki}+\theta_{ki}\Rd_{kj}\right)
 \quad (1\leq i<j\leq r).
\end{align*}

\begin{lemma}
\label{lem:st-diff-operators}
 The above operators annihilate $c(\Theta)$ of $\stiefelrp$.
\end{lemma}

\begin{proof}
 The proof is similar to that of Lemma~\ref{lemma:sogp-annihilate}.
 The operator $A_{ij}^{(1)}$ annihilates $\etr(\Theta^{\top}X)$
 if $X\in \stiefelrp$.
 Since $c(\Theta)=c(Q\Theta)=c(\Theta R)$ for any $Q\in\ogp$ and $R\in\ogr$,
 we have $A_{ij}^{(2)}\cdot c(\Theta)=0$ and $\tilde{A}^{(2)}\cdot c(\Theta)=0$,
 respectively.
\end{proof}

Let $I$ denote the ideal generated by the above operators and 
let $I_{\rm diag}$ denote its restriction to diagonal matrices
$\Theta=\mathrm{diag}(\theta_{11},\ldots,\theta_{rr})\in\mbR^{p\times r}$.
We denote by $R_{r,p}$ the ring of differential operators with rational
function coefficients in $\theta_{ij}$, $1 \leq i \leq p$, $1 \leq j \leq r$.

\begin{proposition}
If $r=2$, $p=3$, then the ideal $R_{2,3}I$ is zero-dimensional.
The holonomic rank is equal to $4$.
$R_{2,3}/(R_{2,3}I)$ is spanned by $1, \partial_{11}, \partial_{12}, \partial_{11}^2$
over the field of rational functions.
\end{proposition}
This proposition is also proved by a computation on Risa/Asir.
Programs to verify the proposition are at the website OpenXM/Math (\cite{openxmmath}).
We conjecture that $I$ is holonomic and consequently
$R_{r,p} I$ is zero-dimensional
for any $r$ and $p$ in the case of $\stiefelrp$.

\subsubsection{Differential equations for the diagonal matrix}

For the hypergeometric function
$c(\Theta)={}_0F_1(p/2,Y)$, $Y=\Theta^\top \Theta/4$, 
the following partial differential
equation is well known (\cite{Muirhead1970}, \cite[Thm.7.5.6]{muirhead82}).
Let $y_1,\ldots,y_r$ denote the eigenvalues of $Y$.
The function ${}_0F_1$ satisfies the following partial differential equations:
\begin{equation}
\label{eq:hyper-PDE}
 y_i\Rd_i^2F
 +\left\{\frac{p}{2}-\frac{r-1}{2}+\frac{1}{2}\sum_{j=1,j\neq i}^r\frac{y_i}{y_i-y_j}\right\}\Rd_iF
 -\frac{1}{2}\sum_{j=1,j\neq i}^{r}\frac{y_j}{y_i-y_j}\Rd_jF\ =\ F, \quad i=1,\dots,r.
\end{equation}
\cite{Muirhead1970} 
obtained these partial differential equations from 
the partial differential equations satisfied by zonal polynomials
(\cite{James1968}, \cite[Sec.4.5]{takemura-zonal}).
In Appendix \ref{app:pde-hypergeometric}
we check that for low dimensional cases 
these equations are also derived from 
the differential operators in Lemma \ref{lem:st-diff-operators}.

When $\Theta={\rm diag}(\theta_{ii})$ is diagonal,
the normalizing constant $c(\Theta)$ satisfies 
Muirhead's differential equation given below in (\ref{eqn:Muirhead}).
The holonomic rank of the system of equations is $8$ when $p=r=3$.
In the case of $SO(3)$, which is of interest in our applications,
the normalizing constant should satisfy extra partial differential
equations,
because we have shown that the holonomic rank of $R_3 I$ 
is less than or equal to $4$
in Proposition \ref{prop:2}.
In fact, we can find extra differential equations  
from $A_{ij}^{(k)}$ and ${\tilde A}_{ij}^{(k)}$.

\begin{theorem}  \label{th:extraeq}
\begin{enumerate}
\item
Put
\begin{equation} \label{eq:extraeq}
\ell_{ij} = (\theta_{ii}^2-\theta_{jj}^2)\partial_{ii}\partial_{jj} 
    - (\theta_{jj}\partial_{ii} - \theta_{ii}\partial_{jj}) - (\theta_{ii}^2-\theta_{jj}^2)\partial_{kk}
\end{equation}
for $i \not= j$. 
The index $k$ is chosen so that $\{i,j,k\} = \{1,2,3\}$.
Then, the normalizing constant $c({\rm diag}(\theta_{ii}))$ 
satisfies the partial differential equation
$\ell_{ij} \cdot c = 0$
for $1 \leq i < j \leq 3$.
\item The holonomic rank of (\ref{eqn:Muirhead})  and 
(\ref{eq:extraeq}), $ij=12,13,23$, is $4$.
\end{enumerate}
\end{theorem}

In order to find the operator $\ell_{ij}$,
we utilized the series expansion (\ref{eq:ctheta-expansion}) and (\ref{eq:eklm})
of the normalizing constant $c$,
which will be given in the next subsection,
and the method of undetermined coefficients or a syzygy computation.
These methods will be explained after the proof.
Once these operators $\ell_{ij}$ and some auxiliary operators are found,
the proof consists of a tedious calculation.

\begin{proof}
1.
Let $I$ be the left ideal in $D_9$, which is the ring of differential
operators in $9$ variables, generated by
$A^{(k)}_{ij}$, $\tilde{A}^{(k)}_{ij}$, $A^{(2)}$.
Since the normalizing constant $c$ is annihilated by the elements of $I$,
we may show that
$\ell_{ij} \in I + \sum_{i \not= j} \theta_{ij} D_9$.
In fact, we have
\begin{eqnarray*}
\ell_{12} &=& 
a_{13} A^{(1)}_{13} + a_{32} A^{(1)}_{32} 
+ a_{33} A^{(1)}_{33} + b A^{(2)} + c_{21} A^{(3)}_{21} 
+ \tilde{c}_{21} \tilde{A}^{(3)}_{21}
+ \sum_{i\ne j} \theta_{ij} e_{ij}
\end{eqnarray*} 
where
\begin{eqnarray*}
a_{13} &=& (\theta_{22}^2-\theta_{11}^2)(\partial_{21}\partial_{32}-\partial_{22}\partial_{31})  \\
a_{32} &=&-(\theta_{22}^2-\theta_{11}^2)(\partial_{11}\partial_{32}-\partial_{12}\partial_{31})  \\
a_{33} &=& (\theta_{22}^2-\theta_{11}^2)(\partial_{11}\partial_{22}-\partial_{12}\partial_{21})  \\
b      &=&-(\theta_{22}^2-\theta_{11}^2)\partial_{33}  \\
{c}_{21} &=& \theta_{22}(\theta_{22}\partial_{21}-\theta_{11}\partial_{12})\partial_{22}  \\
\tilde{c}_{21}&=& (\theta_{11}\partial_{21}-\theta_{22}\partial_{12})(\theta_{22}\partial_{22}-1) - 2\theta_{22}\partial_{12}
\end{eqnarray*}
and
\begin{eqnarray*}
e_{12} &=&
\theta_{22}((\theta_{22} \partial_{22}+1) \partial_{11}-\theta_{11} \partial_{22}^2) \partial_{12}
-\theta_{11} (\theta_{22} \partial_{22}-1)\partial_{11}\partial_{21} 
+\theta_{22}^2 \partial_{22}^2 \partial_{21}
\\
e_{13} &=&
\theta_{22} (\theta_{22}\partial_{21}-\theta_{11}\partial_{12}) \partial_{22}\partial_{23}
\\
e_{21} &=&
(\theta_{22} \partial_{22} \theta_{11} \partial_{11}-\theta_{22}^2 \partial_{22}^2-\theta_{22} \partial_{22}) \partial_{12}
-\theta_{22}^2 \partial_{22} \partial_{21} \partial_{11}
+\theta_{11}(\theta_{22} \partial_{22}-1) \partial_{21}\partial_{22}
\\
e_{23} &=&
-\theta_{22} (\theta_{22} \partial_{21}-\theta_{11} \partial_{12}) \partial_{13} \partial_{22}
\\
e_{31} &=&
(\theta_{11}(\theta_{22} \partial_{22}- 1) \partial_{21} -\theta_{22}(\theta_{22} \partial_{22}+ 1) \partial_{12}) \partial_{32}
\\
e_{32} &=&
(\theta_{22} (\theta_{22} \partial_{22}+1 )\partial_{12} - \theta_{11}(\theta_{22} \partial_{22}-1) \partial_{21}) \partial_{31}.
\end{eqnarray*}
We can show that 
the differential operator (\ref{eqn:Muirhead})
found by the Muirhead also belongs to $I + \sum_{i \not= j} \theta_{ij} D_9$
with an analogous method.
Now, the second statement can be shown by a rank evaluation program
(use, e.g., the {\tt holonomicRank} command of Macaulay2).
\end{proof}


Let us explain how we found the operator $\ell_{12}$.
We put 
$\ell_{12} = \sum d_{ij}(\theta') \partial_{ii} \partial_{jj}
            +\sum d_{i}(\theta') \partial_{ii}
            +d(\theta')
$, $\theta'=(\theta_{11},\theta_{22},\theta_{33})$
where the degree of polynomials $d_{ij}, d_i, d$ is less than or equal to $2$.
We act the operator $\ell_{12}$ to a truncated series expansion and
obtain a system of linear equations for the undetermined coefficients of the polynomials.
By solving the system, we get a candidate of $\ell_{12}$.
The operators $a_{ij}$, ${\tilde a}_{ij}$, $b$, ${\tilde c}_{ij}$ 
and $c_{ij}$ are also found by the method 
of undetermined coefficients.
More precisely speaking,
we put $a_{ij} = \sum_{|\alpha|+|\beta|\leq N} c_{ij}^{\alpha \beta} \theta^\alpha \partial^\beta$
where $c_{ij}^{\alpha \beta}$ are undetermined coefficients
and $N=5$.
We put other operators analogously.
Expand 
\begin{equation} \label{eq:xright}
\ell_{12}+\sum a_{ij} A^{(1)}_{ij}
 + \sum {\tilde a}_{ij} {\tilde A}^{(1)}_{ij}
 + b A^{(2)} 
 + \sum c_{ij} A^{(3)}_{ij} 
 + \sum {\tilde c}_{ij} {\tilde A}^{(3)}_{ij} 
\end{equation}
into the normally ordered expression 
and put $\theta_{ij}=0$ for all $i\not= j$.
And then, set the coefficients of each $\theta^{\alpha} \partial^\beta$ to $0$
and we obtain a system of linear equations with respect to the undetermined
coefficients.
Find a non-trivial solution of the system which gives a candidate of undetermined operators.
The operators $e_{ij}$ are obtained by collecting the right coefficients
of (\ref{eq:xright}) with respect to $\theta_{ij}$, $i \not= j$.

A different approach is
a syzygy computation (see, e.g., \cite{CLO2}).
We put $a_{ij} = \sum_{|\beta|\leq N} c_{ij}^\beta \partial^\beta$
were $c_{ij}^\beta$ are undetermined polynomials.
We put other operators analogously.
Doing the same procedure as above,
we obtain a system of linear indefinite equations in the polynomial ring
${\bf Q}[\theta_{11}, \theta_{22}, \theta_{33}]$.
It can be solved by the syzygy computation.
The performance of the second method is more efficient than the first
one.

\begin{theorem}  \label{th:pfaffianSO3}
The Pfaffian equation derived from (\ref{eqn:Muirhead})
 and (\ref{eq:extraeq}), $ij=12,13,23$ is 
\begin{equation}  \label{eq:pfaffianSO3}
\partial_{ii}C = P_i C, 
\end{equation}
where
$C=(c,\partial_{11}c,\partial_{22}c,\partial_{33}c)^T$
and
$$P_1 = 
\begin{pmatrix}
0&  1& 0& 0 \\
 1&  \frac{    -    \theta_{11}  (    2   \theta_{11}^{ 2} -  \theta_{22}^{ 2} -  \theta_{33}^{ 2} )} {    (  \theta_{11}- \theta_{33})  (  \theta_{11}+ \theta_{33})  (  \theta_{11}- \theta_{22})  (  \theta_{11}+ \theta_{22})}&  \frac{  \theta_{22}} {  (  \theta_{11}- \theta_{22})  (  \theta_{11}+ \theta_{22})}&  \frac{  \theta_{33}} {  (  \theta_{11}- \theta_{33})  (  \theta_{11}+ \theta_{33})} \\
0&  \frac{  \theta_{22}} {  (  \theta_{11}- \theta_{22})  (  \theta_{11}+ \theta_{22})}&  \frac{   -  \theta_{11}} {  (  \theta_{11}- \theta_{22})  (  \theta_{11}+ \theta_{22})}&  1 \\
0&  \frac{  \theta_{33}} {  (  \theta_{11}- \theta_{33})  (  \theta_{11}+ \theta_{33})}&  1&  \frac{   -  \theta_{11}} {  (  \theta_{11}- \theta_{33})  (  \theta_{11}+ \theta_{33})} \\
\end{pmatrix}
$$
$$P_2 = 
\begin{pmatrix}
0& 0&  1& 0 \\
0&  \frac{  \theta_{22}} {  (  \theta_{11}- \theta_{22})  (  \theta_{11}+ \theta_{22})}&  \frac{   -  \theta_{11}} {  (  \theta_{11}- \theta_{22})  (  \theta_{11}+ \theta_{22})}&  1 \\
 1&  \frac{   -  \theta_{11}} {  (  \theta_{11}- \theta_{22})  (  \theta_{11}+ \theta_{22})}&  \frac{    -  \theta_{22}  (    \theta_{11}^{ 2} -  2   \theta_{22}^{ 2} +  \theta_{33}^{ 2} )} {    (  \theta_{22}- \theta_{33})  (  \theta_{22}+ \theta_{33})  (  \theta_{11}- \theta_{22})  (  \theta_{11}+ \theta_{22})}&  \frac{  \theta_{33}} {  (  \theta_{22}- \theta_{33})  (  \theta_{22}+ \theta_{33})} \\
0&  1&  \frac{  \theta_{33}} {  (  \theta_{22}- \theta_{33})  (  \theta_{22}+ \theta_{33})}&  \frac{   -  \theta_{22}} {  (  \theta_{22}- \theta_{33})  (  \theta_{22}+ \theta_{33})} \\
\end{pmatrix}
$$
$$P_3 = 
\begin{pmatrix}
0& 0& 0&  1 \\
0&  \frac{  \theta_{33}} {  (  \theta_{11}- \theta_{33})  (  \theta_{11}+ \theta_{33})}&  1&  \frac{   -  \theta_{11}} {  (  \theta_{11}- \theta_{33})  (  \theta_{11}+ \theta_{33})} \\
0&  1&  \frac{  \theta_{33}} {  (  \theta_{22}- \theta_{33})  (  \theta_{22}+ \theta_{33})}&  \frac{   -   \theta_{22}} {  (  \theta_{22}- \theta_{33})  (  \theta_{22}+ \theta_{33})} \\
 1&  \frac{   -  \theta_{11}} {  (  \theta_{11}- \theta_{33})  (
 \theta_{11}+ \theta_{33})}&  \frac{  -  \theta_{22}} {  (  \theta_{22}-
 \theta_{33})  (  \theta_{22}+ \theta_{33})}&  \frac{   \theta_{33}  (
 \theta_{11}^{ 2} +  \theta_{22}^{ 2} -  2   \theta_{33}^{ 2} )} {    (
 \theta_{22}- \theta_{33})  (  \theta_{22}+ \theta_{33})  (  \theta_{11}-
 \theta_{33})  (  \theta_{11}+ \theta_{33})} \\
\end{pmatrix}
$$
\end{theorem}

This theorem can be shown by a straightforward calculation
from Theorem \ref{th:extraeq}
as explained in, e.g., \cite{hgd}.

We evaluate the normalizing constant and its derivatives
by evaluating a truncated series expansion 
near the origin and extend the value by solving an ordinary
differential equation (the holonomic gradient method, \cite{HNTT}).
The ODE is given in the following Corollary.
\begin{corollary}
For constants $a,b,c$, we restrict the function $C$ to
$\theta_{11} = at$, $\theta_{22}=bt$, $\theta_{33}=ct$.
The ordinary differential equation satisfied by $C$ with respect to $t$
is
$$ \frac{dC}{dt} = 
\left( A - \frac{2}{t}\, {\rm diag}(0,1,1,1) \right) C, 
$$
where
$$ A = 
\begin{pmatrix}
 0 & a & b & c \\
 a & 0 & c & b \\
 b & c & 0 & a \\
 c & b & a & 0 \\
\end{pmatrix}
$$
and the eigenvalues of $A$ are
$\{ a-b-c,-a+b-c,-a-b+c,a+b+c \}$.
\end{corollary}


\subsubsection{Practice of HGD} \label{sec:grid}
Although the HGD is a general method
which can be applied to broad problems, 
we need a good guess (oracle) of a starting point to search for
the optimal point (MLE).

We explain why we need a good guess of a starting point with an example
of $SO(3)$.
Let $\Theta$ be the optimal point for a given data
and $\frac{\partial C}{\partial \theta_{ii}} = P_{i}(\theta) C$
be the Pfaffian system to apply for the HGD.
The denominator of the coefficient matrix $P_{ij}$ is the polynomial
$\prod_{1\leq i < j \leq 3} (\theta_{ii} \pm \theta_{jj})$
by Theorem \ref{th:pfaffianSO3}. 
We denote by $V$ the zero set of the polynomial.
The numerical integration procedure
of the Pfaffian system becomes unstable near the zero
set $V$, which is called the singular locus
of the Pfaffian system.
Therefore, the starting point must be in the same component as the optimal point
in $\mbR^3 \setminus V$. 
In our current implementation of HGD, we have three heuristic methods
to find a starting point:
\begin{enumerate}
\item
We find the starting point
by preparing a table of the values of the normalizing constant 
at grids and making the exhaustive search of the optimal point on the grids.
Note that the table of the normalizing constant does not depend on specific data.
\item 
In the case of $SO(3)$, we have $24$
connected components in $\mbR^3
      \setminus V$.
We choose starting points in the  $24$ connected components
and apply the HGD for these points until a success.
\item Solve the MLE problem by an approximate method. Use the approximate 
value of the MLE
as a starting point.
For the Fisher distribution, there are several methods to find approximate value of the MLE.
They are (1) series expansion method given in this paper
and (2) the method by \cite{wood-1993}.
\end{enumerate}
As a referee pointed out,
the non-parametric estimation method by \cite{Beran1979} can also be used to give a starting point.

%
%

\subsection{Series expansion approach for $\sog$ and $\stiefel$}
\label{subsec:series-expansion}

We describe a method to compute the maximum likelihood estimate
by an infinite series expansion of $c(\Theta)$.
By Lemma~\ref{lemma:mle-diag},
computation of the maximum likelihood estimate for $\sogp$
is reduced to computation of $c(\diag(\phi_1,\ldots,\phi_p))$ and its derivatives
with respect to $\phi_i$'s, together with the usual gradient method.
In this subsection we give an explicit series expansion of $c(\diag(\phi_1,\phi_2,\phi_3))$ when $p=3$.
Note that $c(\Theta)$ for any $\Theta\in\mathbb{R}^{3\times 3}$ is
also obtained via sign-preserving SVD due to the rotational invariance of $c(\Theta)$.
By using the expansion formula we also clarify
the difference between the normalizing constants for the orthogonal group
$\og$ and the special orthogonal group $\sog$.
The series expansion approach for $\stiefel$ is also discussed.

From mathematical viewpoint, the holonomic gradient descent and the infinite
series expansion are related as follows.
In the general recipe of the holonomic gradient descent and holonomic systems,
we can construct series expansion of the normalizing constant 
$c(\Theta)$ 
up to any degree modulo finite constants by an algorithmic method
from a holonomic system of differential equations satisfied by $c(\Theta)$,
which is obtained in the previous subsection.
The existence of finite recurrence relations for coefficients of the series
is proved by the holonomicity.
This is a multi-variable generalization of the fact that coefficients
of series solutions of linear ODE satisfy a finite recurrence relation.
Since this computation requires huge computational resources,
constructing the series expansion in a more efficient way is preferable 
to using the general algorithm.
Here we derive an infinite series expansion for $\sog$ with an analysis of integrals. 

Let $E[\cdot]$ denote the expectation with respect to the uniform
distribution on $\sog$.
Let $\phi_1, \phi_2, \phi_3$ be the sign-preserving singular values of $\Theta$.
By the rotational invariance,
the expansion of $c(\Theta)$ is
\begin{align}
c(\Theta)&=\sum_{h=0}^\infty \frac{1}{h!} E[(\tr \Theta^\top X)^h]
= \sum_{h=0}^\infty \frac{1}{h!} E[(\phi_1 x_{11} + \phi_2 x_{22} + \phi_3 x_{33})^h]
\nonumber \\
&=\sum_{k,l,m=0}^\infty \frac{1}{k! \, l! \, m!} \phi_1^k \phi_2^l \phi_3^m E[x_{11}^k x_{22}^l x_{33}^m]
\label{eq:ctheta-expansion}
\end{align}
and the problem is reduced to the evaluation of
\[
E(k,l,m)=E[x_{11}^k x_{22}^l x_{33}^m].
\]
Again by the rotational invariance
we can simultaneously change the sign
of any two of $x_{11}, x_{22}, x_{33}$.   From this it is easily seen
that $E(k,l,m)=0$ unless $k,l,m$ are all even or $k,l,m$ are all
odd.

Note that for $\og$ we can individually change the signs of 
$x_{11}, x_{22}, x_{33}$.  Hence for $\og$
$E(k,l,m)=0$ unless $k,l,m$ are all even and 
$c(\Theta)$ is indeed a function of the eigenvalues of $Y=\Theta^\top \Theta/4$.
Therefore the difference between  $c(\Theta)$ for $\sog$ and $c(\Theta)$ for $\og$ 
comes from terms 
$E[k,l,m]=0$ with $k,l,m$ all odd.

We now express $X = (x_{ij})\in \sog$ by the Euler angles $\theta,\phi,\psi$.
\[ 
X = 
{\footnotesize
\begin{pmatrix}
\sin\theta\sin\phi  &
\hphantom{-}\cos\phi\sin\psi + \cos\theta\sin\phi\cos\psi 
&
-\cos\phi\cos\psi + \cos\theta\sin\phi\sin\psi \\
\\
\sin\theta\cos\phi &
-\sin\phi\sin\psi + \cos\theta\cos\phi\cos\psi 
&
\hphantom{-}\sin\phi\cos\psi + \cos\theta\cos\phi\sin\psi \\
\\
\cos\theta\hphantom{\cos\phi} 
& -\sin\theta\cos\psi
& -\sin\theta\sin\psi
\end{pmatrix}.
}
\]
The Jacobian of the above transformation is $\sin\theta$ and the
range of variables is
\[ 0 \le \theta \le \pi,\quad 0 \le \phi \le 2 \pi,\quad 0 \le \psi \le 2 \pi. \]
Hence the integral of $f$ over $\sog$ with respect to the uniform distribution is expressed as
\[
\int_{\sog} f(X) d \mu(X)
= \frac{1}{8\pi^2}\int_0^{\pi} d\theta\int_0^{2\pi} d\phi\int_0^{2\pi} d\psi \;
f (X(\theta,\phi,\psi)) \sin\theta.
\]

For 
\[ f = x_{11}^kx_{22}^lx_{33}^m = (\sin\theta\sin\phi)^k(-\sin\phi\sin\psi + \cos\theta\cos\phi\cos\psi)^l(-\sin\theta\sin\psi)^m
 \]
we have
\begin{align*}
f\cdot \sin\theta &= 
(-1)^m\sin^{k+m+1}\theta\sin^k\phi\sin^m\psi \\
& \qquad \qquad \cdot
\sum_{n=0}^l \binom{l}{n} (-1)^n \sin^n\phi\sin^n\psi\cos^{l-n}\theta\cos^{l-n}\phi\cos^{l-n}\psi \\
&= \sum_{n=0}^l \binom{l}{n} (-1)^{m+n}\sin^{k+m+1}\theta\cos^{l-n}\theta \sin^{k+n}\phi \cos^{l-n}\phi \sin^{m+n}\psi \cos^{l-n}\psi  .
\end{align*}

Define
\[ I[m,n] = \frac{(m-1)!! (n-1)!! }{(m+n)!!}, \]
where $(2a)!!=\prod_{j=1}^a(2j)$ and $(2a-1)!!=\prod_{j=1}^a(2j-1)$
for each non-negative integer $a$.
Then from well-known results on the definite integrals of trigonometric functions
we have
\begin{equation}
\label{eq:eklm}
E(k,l,m) = \sum_{0\le n\le l \> \atop 
l-n:  \text{ even} } \binom{l}{n}
I[k+m+1,l-n]\cdot I[k+n,l-n] \cdot I[m+n,l-n] .
\end{equation}
By numerical experiments we found that  (\ref{eq:eklm}) can be computed easily
and we can evaluate $c(\Theta)$ by the right-hand side of
(\ref{eq:ctheta-expansion}) to a desired accuracy.
For large $k,l,m$ the value of $E(k,l,m)$ can be approximated by
Laplace's method.  Laplace approximation to $E(k,l,m)$ is given in  Appendix
\ref{app:eklm}.


We now consider the maximization of (\ref{eq:objective}) with respect to 
$\{\phi_i\}_{i=1}^3$ when we adopt direct use of the gradient descent.
The gradient method uses the first derivatives of (\ref{eq:objective}).
The Hessian matrix is also needed if one uses the Newton method.
Since the first term of (\ref{eq:objective}) is linear,
it is sufficient to give the series expansion of
the derivatives of $c(\diag(\phi_1,\phi_2,\phi_3))$.
They are easily obtained from the expansion of $c(\Theta)$.
For example the derivative with respect to $\phi_1$ is
\begin{align*}
 \frac{\partial c(\diag(\phi_1,\phi_2,\phi_3))}{\partial\phi_1}
 = \sum_{k,l,m=0}^{\infty}\frac{1}{k! \, l! \, m!}\phi_1^k \phi_2^l 
 \phi_3^m E(k+1,l,m).
\end{align*}
Similarly,
\begin{align*}
& \frac{\partial^2c(\diag(\phi_1,\phi_2,\phi_3))}{\partial\phi_1^2}
=\sum_{k,l,m=0}^\infty \frac{1}{k! \, l! \, m!}  \phi_1^k \phi_2^l \phi_3^m E(k+2,l,m),
\\
& \frac{\partial^2c(\diag(\phi_1,\phi_2,\phi_3))}{\partial\phi_1\partial\phi_2}
=\sum_{k,l,m=0}^\infty \frac{1}{k! \, l! \, m!}  \phi_1^k \phi_2^l \phi_3^m E(k+1,l+1,m).
\end{align*}

Finally we note that the series expansion of $c(\Theta)$ for $\sog$
is directly used for the maximum likelihood estimate of $\stiefel$.
Let $\bar X_{1:2}$ be the first two columns of
the averaged data matrix $\bar X\in\mathbb{R}^{3\times 3}$.
Let $\bar X_{1:2}=Q\diag(g_1,g_2)R$ be the (usual) SVD.
Then, as stated before Lemma~\ref{lemma:mle-diag},
the maximum likelihood estimator for $\stiefel$ is
given by $\hat\Theta=Q\diag(\hat\phi_1,\hat\phi_2)R$, where
$(\hat\phi_i)$ is the maximizer of
\begin{align*}
 \sum_{k=1}^2 \phi_kg_k - \log\left(\int_{\stiefel}\exp(\sum_{k=1}^2\phi_kx_{kk})\mu(dX)\right)
 = \sum_{k=1}^2 \phi_kg_k - \log c(\diag(\phi_1,\phi_2,0))
\end{align*}
in terms of $c(\Theta)$ for $\sog$.
Then the MLE is obtained via the series expansion of $c(\Theta)$.

\section{Application to data on orbits of near-earth objects}
\label{sec:data-analysis}

In this section as an illustration of the above discussion, 
we fit Fisher distributions of $\sog$ and  $\stiefel$ 
to data of  orbits  of near-earth objects.
We obtained the data from the web page of
Near Earth Object Program of National Aeronautics and Space Administration
(cf.\ \verb@http://neo.jpl.nasa.gov/cgi-bin/neo_elem@).  
Near-earth objects are comets and asteroids around the Earth.  
\cite{JuppMardia1979} fitted
Fisher distribution on $\stiefel$ to data of comets from \cite{marsden-catalog-1972}, but
did not consider Fisher distribution on $\sog$.  
See also \cite{Mardia1975} for analysis of data of perihelion direction.

The near-earth objects  have ellipsoidal orbits with the Sun as their focus.  The
orbits are characterized by the following two directions:
\begin{enumerate}
\setlength{\itemsep}{0pt}
\item the perihelion direction $\bx_1$, which is the direction of the
closest point on the orbit from the Sun.
\item the normal direction $\bx_2$ to the orbit, which is determined by the 
right-hand rule for the rotation of the object.
\end{enumerate}
The pair $(\bx_1,\bx_2)$ is an element of $\stiefel$.
We can also define $\bx_3=\bx_1\times \bx_2$ 
such that $(\bx_1, \bx_2, \bx_3)$  is an element of $SO(3)$.
\begin{figure}[htbp]
\begin{center}
\unitlength 0.1in
\begin{picture}( 42.2800, 24.2000)(  0.9000,-28.3000)
%
\special{pn 8}%
\special{pa 2882 910}%
\special{pa 2894 940}%
\special{pa 2906 970}%
\special{pa 2912 1002}%
\special{pa 2916 1034}%
\special{pa 2914 1066}%
\special{pa 2910 1098}%
\special{pa 2904 1130}%
\special{pa 2896 1160}%
\special{pa 2886 1190}%
\special{pa 2876 1220}%
\special{pa 2864 1250}%
\special{pa 2850 1280}%
\special{pa 2834 1308}%
\special{pa 2818 1336}%
\special{pa 2802 1362}%
\special{pa 2784 1390}%
\special{pa 2766 1416}%
\special{pa 2748 1442}%
\special{pa 2728 1468}%
\special{pa 2708 1494}%
\special{pa 2688 1518}%
\special{pa 2668 1542}%
\special{pa 2648 1568}%
\special{pa 2626 1592}%
\special{pa 2606 1614}%
\special{pa 2584 1638}%
\special{pa 2562 1662}%
\special{pa 2538 1684}%
\special{pa 2516 1708}%
\special{pa 2494 1730}%
\special{pa 2470 1752}%
\special{pa 2446 1774}%
\special{pa 2422 1794}%
\special{pa 2398 1816}%
\special{pa 2374 1836}%
\special{pa 2350 1858}%
\special{pa 2326 1878}%
\special{pa 2302 1898}%
\special{pa 2276 1918}%
\special{pa 2252 1938}%
\special{pa 2226 1958}%
\special{pa 2200 1978}%
\special{pa 2174 1996}%
\special{pa 2148 2016}%
\special{pa 2124 2034}%
\special{pa 2098 2054}%
\special{pa 2070 2072}%
\special{pa 2044 2090}%
\special{pa 2018 2108}%
\special{pa 1992 2126}%
\special{pa 1964 2142}%
\special{pa 1938 2160}%
\special{pa 1910 2178}%
\special{pa 1884 2194}%
\special{pa 1856 2210}%
\special{pa 1828 2228}%
\special{pa 1800 2244}%
\special{pa 1774 2260}%
\special{pa 1746 2276}%
\special{pa 1718 2290}%
\special{pa 1690 2306}%
\special{pa 1662 2322}%
\special{pa 1634 2336}%
\special{pa 1604 2352}%
\special{pa 1576 2366}%
\special{pa 1548 2380}%
\special{pa 1518 2394}%
\special{pa 1490 2408}%
\special{pa 1462 2422}%
\special{pa 1432 2436}%
\special{pa 1402 2448}%
\special{pa 1374 2462}%
\special{pa 1344 2474}%
\special{pa 1314 2486}%
\special{pa 1286 2500}%
\special{pa 1256 2512}%
\special{pa 1226 2522}%
\special{pa 1196 2534}%
\special{pa 1166 2546}%
\special{pa 1136 2556}%
\special{pa 1106 2566}%
\special{pa 1076 2576}%
\special{pa 1044 2586}%
\special{pa 1014 2596}%
\special{pa 984 2604}%
\special{pa 952 2612}%
\special{pa 922 2622}%
\special{pa 890 2628}%
\special{pa 860 2636}%
\special{pa 828 2644}%
\special{pa 796 2650}%
\special{pa 766 2656}%
\special{pa 734 2662}%
\special{pa 702 2668}%
\special{pa 670 2672}%
\special{pa 638 2676}%
\special{pa 606 2680}%
\special{pa 576 2682}%
\special{pa 544 2684}%
\special{pa 512 2684}%
\special{pa 480 2684}%
\special{pa 448 2682}%
\special{pa 416 2680}%
\special{pa 384 2676}%
\special{pa 352 2670}%
\special{pa 320 2664}%
\special{pa 290 2654}%
\special{pa 260 2644}%
\special{pa 230 2630}%
\special{pa 202 2614}%
\special{pa 176 2594}%
\special{pa 154 2572}%
\special{pa 134 2548}%
\special{pa 116 2520}%
\special{pa 104 2490}%
\special{pa 96 2458}%
\special{pa 92 2426}%
\special{pa 90 2394}%
\special{pa 94 2362}%
\special{pa 98 2332}%
\special{pa 106 2300}%
\special{pa 114 2270}%
\special{pa 126 2238}%
\special{pa 138 2210}%
\special{pa 150 2180}%
\special{pa 164 2152}%
\special{pa 180 2122}%
\special{pa 196 2094}%
\special{pa 212 2068}%
\special{pa 230 2042}%
\special{pa 248 2014}%
\special{pa 268 1990}%
\special{pa 286 1964}%
\special{pa 306 1938}%
\special{pa 328 1914}%
\special{pa 348 1890}%
\special{pa 368 1866}%
\special{pa 390 1842}%
\special{pa 412 1818}%
\special{pa 434 1794}%
\special{pa 456 1772}%
\special{pa 478 1750}%
\special{pa 502 1726}%
\special{pa 524 1704}%
\special{pa 548 1682}%
\special{pa 572 1660}%
\special{pa 596 1640}%
\special{pa 620 1618}%
\special{pa 644 1598}%
\special{pa 668 1576}%
\special{pa 692 1556}%
\special{pa 718 1536}%
\special{pa 742 1516}%
\special{pa 768 1496}%
\special{pa 792 1476}%
\special{pa 818 1458}%
\special{pa 844 1438}%
\special{pa 870 1420}%
\special{pa 896 1400}%
\special{pa 922 1382}%
\special{pa 948 1364}%
\special{pa 974 1346}%
\special{pa 1002 1328}%
\special{pa 1028 1310}%
\special{pa 1054 1292}%
\special{pa 1082 1276}%
\special{pa 1108 1258}%
\special{pa 1136 1242}%
\special{pa 1164 1226}%
\special{pa 1190 1208}%
\special{pa 1218 1192}%
\special{pa 1246 1176}%
\special{pa 1274 1162}%
\special{pa 1302 1146}%
\special{pa 1330 1130}%
\special{pa 1358 1116}%
\special{pa 1386 1100}%
\special{pa 1416 1086}%
\special{pa 1444 1072}%
\special{pa 1472 1056}%
\special{pa 1502 1044}%
\special{pa 1530 1030}%
\special{pa 1560 1016}%
\special{pa 1588 1002}%
\special{pa 1618 990}%
\special{pa 1646 976}%
\special{pa 1676 964}%
\special{pa 1706 952}%
\special{pa 1736 940}%
\special{pa 1764 928}%
\special{pa 1794 916}%
\special{pa 1824 904}%
\special{pa 1854 894}%
\special{pa 1886 882}%
\special{pa 1916 872}%
\special{pa 1946 864}%
\special{pa 1976 854}%
\special{pa 2008 844}%
\special{pa 2038 836}%
\special{pa 2068 828}%
\special{pa 2100 818}%
\special{pa 2130 810}%
\special{pa 2162 804}%
\special{pa 2194 796}%
\special{pa 2224 790}%
\special{pa 2256 784}%
\special{pa 2288 778}%
\special{pa 2320 774}%
\special{pa 2350 770}%
\special{pa 2382 766}%
\special{pa 2414 764}%
\special{pa 2446 762}%
\special{pa 2478 760}%
\special{pa 2510 760}%
\special{pa 2542 760}%
\special{pa 2574 762}%
\special{pa 2606 766}%
\special{pa 2638 770}%
\special{pa 2670 776}%
\special{pa 2700 784}%
\special{pa 2732 794}%
\special{pa 2762 806}%
\special{pa 2790 820}%
\special{pa 2816 838}%
\special{pa 2840 860}%
\special{pa 2862 884}%
\special{pa 2882 910}%
\special{pa 2882 910}%
\special{sp}%
%
\special{pn 20}%
\special{sh 1}%
\special{ar 2614 1068 10 10 0  6.28318530717959E+0000}%
%
\special{pn 8}%
\special{pa 518 1070}%
\special{pa 4318 1070}%
\special{da 0.070}%
%
\special{pn 8}%
\special{pa 2618 410}%
\special{pa 2618 2010}%
\special{da 0.070}%
%
\special{pn 8}%
\special{pa 2838 570}%
\special{pa 2038 2370}%
\special{da 0.070}%
%
\special{pn 8}%
\special{pa 1498 2310}%
\special{pa 1698 2290}%
\special{fp}%
\special{pa 1698 2290}%
\special{pa 1628 2450}%
\special{fp}%
%
\special{pn 13}%
\special{pa 2608 1060}%
\special{pa 2948 830}%
\special{fp}%
\special{sh 1}%
\special{pa 2948 830}%
\special{pa 2882 852}%
\special{pa 2904 860}%
\special{pa 2904 884}%
\special{pa 2948 830}%
\special{fp}%
\special{pa 2618 1060}%
\special{pa 2378 810}%
\special{fp}%
\special{sh 1}%
\special{pa 2378 810}%
\special{pa 2410 872}%
\special{pa 2416 848}%
\special{pa 2440 844}%
\special{pa 2378 810}%
\special{fp}%
\special{pa 2608 1070}%
\special{pa 2778 1270}%
\special{fp}%
\special{sh 1}%
\special{pa 2778 1270}%
\special{pa 2750 1206}%
\special{pa 2744 1230}%
\special{pa 2720 1232}%
\special{pa 2778 1270}%
\special{fp}%
\put(29.5000,-8.2000){\makebox(0,0)[lt]{$\bm{x}_1$ (the perihelion direction)}}%
\put(22.9000,-6.0000){\makebox(0,0)[lt]{$\bm{x}_2$ (the directed unit normal to the orbit)}}%
\put(13.6000,-26.9000){\makebox(0,0)[lt]{$O$: the Sun}}%
\put(28.1000,-12.6000){\makebox(0,0)[lt]{$\bm{x}_3$ (the vector product $\bm{x}_1\times \bm{x}_2$ )}}%
\put(24.2000,-11.0000){\makebox(0,0)[lt]{$O$}}%
%
\end{picture}%
\end{center}
\caption{Orbits of near-earth objects}
\label{fig:neo}
\end{figure}
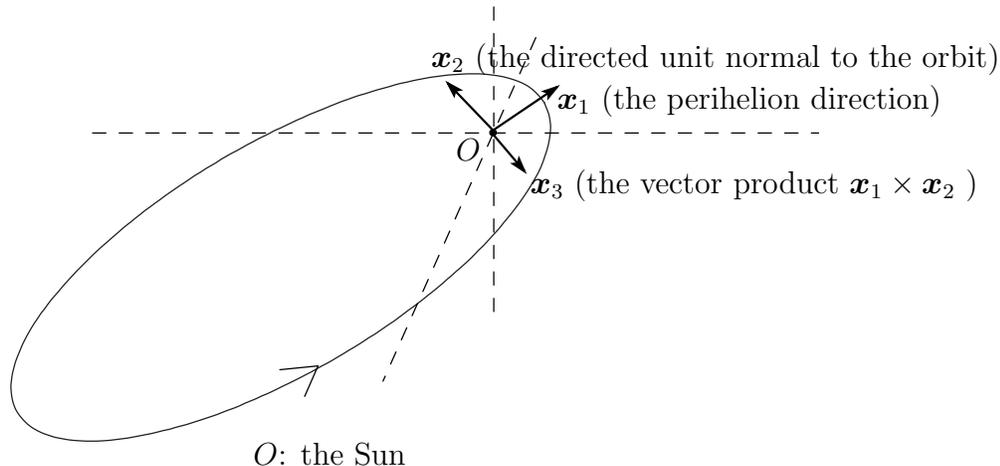

We analyzed 151 comets and 6496 asteroids separately.
To obtain a meaningful result, we identified a tight cluster 
of 67 similar comets, which we treated as one comet, 
and therefore the actual sample size of comets is $N=85$.

The sample mean matrix of the 85 comets is
\begin{align}
 \label{eq:comets-data}
 \bar{\bx} &=
 \begin{pmatrix}
  0.257& \hpm 0.044& \hpm 0.189\\
  0.158&    - 0.052&    - 0.146\\
  0.079& \hpm 0.765& \hpm 0.004
 \end{pmatrix}.
\end{align}
Since the $(3,2)$ element of $\bar{\bm{x}}$ in (\ref{eq:comets-data}) is large,
the orbital plane of the comets are typically close to that of the Earth.
The sample mean matrix of the 6496 asteroids is
\begin{align}
 \label{eq:asteroids-data}
\bar{\bx} &=
 \begin{pmatrix}
\hpm 0.074 & \hpm 0.012 & \hpm 0.016 \\
\hpm 0.018 & \hpm 0.003 & -0.074 \\
-0.001 & \hpm 0.949 & \hpm 0.002
 \end{pmatrix}.
\end{align}

Let $\bar{\bx}_{1:2}$ be the first two columns of $\bar{\bx}$.
We give the SVD of $\bar{\bx}_{1:2}$ and
and the sign-preserving SVD of $\bar{\bx}$.
For the comets data we have
\begin{align}
 \label{eq:comets-data-SVD}
 &\bar{\bx}_{1:2}
 = 
 Q_2
 \begin{pmatrix}
  0.773& 0\\
  0& 0.299
 \end{pmatrix}
 R_2
 \ \ \ \mbox{and}\ \ \ 
 \bar{\bm{x}} =
 Q \begin{pmatrix}
  0.774& 0& 0\\
  0& 0.318& 0\\
  0& 0& 0.210
   \end{pmatrix}
 R,
\end{align}
where
\begin{align*}
 &Q_2= 
 \begin{pmatrix}
  \hpm 0.098& \hpm 0.835\\
  -0.041& \hpm  0.547\\
  \hpm 0.994& -0.060
 \end{pmatrix},
 \ \ \ 
 R_2=\begin{pmatrix}
  0.126& \hpm 0.992\\
  0.992& -0.126
 \end{pmatrix},
 \\
 &
 Q=
 \begin{pmatrix}
    - 0.109& -0.964& -0.242\\
 \hpm 0.048& -0.248&  0.968\\
    - 0.992&  0.093&  0.073
 \end{pmatrix}
 \ \ \ \mbox{and}\ \ \ 
 R=\begin{pmatrix}
      -0.128&     -0.991& -0.041\\
      -0.879& \hpm 0.132& -0.458\\
  \hpm 0.459&     -0.023& -0.888
 \end{pmatrix}.
\end{align*}
For the asteroids data we have
\begin{align}
 \label{eq:asteroids-data-SVD}
 &\bar{\bx}_{1:2}
 = Q_2
 \begin{pmatrix}
  0.949& 0\\
  0& 0.076
 \end{pmatrix}
 R_2
 \ \ \ \mbox{and}\ \ \ 
 \bar{\bx}
 = Q
 \begin{pmatrix}
  0.949& 0& 0\\
  0& 0.0769& 0\\
  0& 0& 0.0749
 \end{pmatrix}
 R,
\end{align}
where
\begin{align*}
 & Q_2=
 \begin{pmatrix}
  -0.0126& \hpm 0.972\\
  -0.00316& \hpm 0.236\\
  -1.000& -0.0130
 \end{pmatrix},
 \ \ \ 
 R_2=\begin{pmatrix}
  7.82\cdot 10^{-6}& -1\\
  1& \hpm 7.82\cdot 10^{-6}
 \end{pmatrix},
 \\
 & Q =
 \begin{pmatrix}
  0.0127&  - 0.623&  0.782 \\
  0.003&  - 0.782&  - 0.623 \\
  1.000&  0.0102&  - 0.00805 \\
 \end{pmatrix}
 \ \ \ \mbox{and}\ \ \ 
 R = \begin{pmatrix}
 - 8.17\cdot 10^{-6}&  1.000&  0.00209 \\
 - 0.782&  - 0.00131&  0.623 \\
  0.623&  - 0.00163&  0.782 \\
\end{pmatrix}.
\end{align*}

As discussed in Section \ref{sec:fisher-distribution}
we can analyze the data either on $\stiefel$ or on $\sog$.
The sufficient statistic of the Fisher distribution on $\stiefel$
is $\bar{\bx}_{1:2}$.

\subsection{The test of uniformity based on Rayleigh's statistic}

As a preliminary analysis we 
test whether the orbits of the comets and asteroids are uniformly distributed
over $\stiefel$ or $\sog$.

We first recall the Rayleigh's statistic for Stiefel manifolds.
Let $\bar{\bx}_{1:r}$ be the sample mean matrix of a data set on $\stiefelrp$
and $N$ be the sample size.
Under the null hypothesis of uniformity over $\stiefelrp$
the Rayleigh statistic 
\begin{equation}
S_{1:r} = pN\cdot \mathrm{tr}(\bar{\bx}_{1:r}^T\bar{\bx}_{1:r})
\label{rayleighstatistic1}
\end{equation}
is asymptotically distributed according to the chi-square distribution
with $rp$ degrees of freedom.
Similarly we can define the Rayleigh statistic for the special orthogonal group.
Let $\bar{\bx}$ be the sample mean matrix of a data set on $\sogp$ and $N$ be the sample size.
Under the null hypothesis of uniformity over $\sogp$,
the Rayleigh statistic
\begin{equation}
S = pN\cdot \mathrm{tr}(\bar{\bx}^T\bar{\bx})
\label{rayleighstatistic2}
\end{equation}
is asymptotically distributed according to the chi-square distribution with $p^2$ degrees of freedom,
whenever $p\geq 3$.
Indeed, $\sqrt{pN}\bar{\bx}$ converges to the $p^2$-dimensional multivariate standard normal distribution
as $N\to\infty$ (see (\ref{eq:sogp-moment})).

Let $\bar{\bx}$ be the comets data (\ref{eq:comets-data})
and $\bar{\bx}_{1:2}$ be the first two columns of $\bar{\bx}$.
The Rayleigh statistic (\ref{rayleighstatistic1}) for $\stiefel$ is
\[
S_{1:2} = 3 \cdot 85 \cdot \mathrm{tr}(\bar{\bx}_{1:2}^{\top}\bar{\bx}_{1:2}) = 175.2
\]
with the $p$-value almost zero.
 The Rayleigh statistic (\ref{rayleighstatistic2}) for $\sog$ is
\[
S =  3\cdot 85 \cdot \mathrm{tr}(\bar{\bx}^{\top}\bar{\bx}) = 189.8
\]
with the $p$-value almost zero.

Similarly for the asteroids data (\ref{eq:asteroids-data}),
the null hypothesis of uniformity is rejected by both 
\[ S_{1:2} = 3\cdot 6496 \cdot \mathrm{tr}(\bar{\bx}_{1:2}^{\top}\bar{\bx}_{1:2}) = 1.77\times 10^4 \]    
and
\[ S = 3 \cdot 6496 \cdot \mathrm{tr}(\bar{\bx}^{\top}\bar{\bx} )= 1.78\times 10^4. \]   
The $p$-values are almost zero.

\subsection{Maximum likelihood estimate  of Fisher distributions}
\label{sec:mle-Fisher}

We compute the MLE (maximum likelihood estimate) 
of the Fisher distribution on $\stiefel$ and $\sog$
by using the two methods described in Section~\ref{sec:two-methods}.
For clarity we denote the parameter of the Fisher distribution
on $\stiefel$ and $\sog$ by $\Theta_{1:2}$ and $\Theta$, respectively.

We first compute the MLE
by using the series expansion approach.
We add a superscript (s) as $\hat\Theta^{\rm (s)}$
for values computed by this method.
For the comets' data the MLE
of the Fisher distribution on $\stiefel$ is
\begin{align*}
 \hat{\Theta}_{1:2}^{\rm (s)}
 &=
 \begin{pmatrix}
  0.689& \hpm 0.341\\
  0.394& -0.229\\
  0.496& \hpm 4.273
 \end{pmatrix}
 = Q_2
 \begin{pmatrix}
  4.326& 0\\
  0& 0.767
 \end{pmatrix}
 R_2
\end{align*}
and the MLE on $\sog$ is
\begin{align*}
 \hat{\Theta}^{\rm (s)}
 &= 
 \begin{pmatrix}
  \hpm 2.953&  \hpm 0.200& 0.871\\
  -0.423& -0.317& 2.390\\
  \hpm 0.378&  \hpm 5.566& 0.251
 \end{pmatrix}
 = 
 Q
 \begin{pmatrix}
  5.614& 0& 0\\
   0& 3.079& 0\\
   0& 0& -2.387\\
 \end{pmatrix}
 R,
\end{align*}
where $Q_2$, $R_2$, $Q$ and $R$ are the matrices in (\ref{eq:comets-data-SVD}).
Note that $\det\bar\bx>0$ but $\det\hat\Theta^{\rm (s)}<0$.
We have approximated the normalizing constant
by our series expansion truncated at the degree $20$
and maximized the approximated log-likelihood function.
These numerical results are obtained by the implementation in R
of the BFGS method.

For the asteroid data the MLEs are
\begin{eqnarray*}
 \hat\Theta^{\rm (s)}_{1:2}
 &=&
 \begin{pmatrix}
  0.156&  0.248 \\
  0.0379&  0.062 \\
  - 0.00225&  19.6 \\
 \end{pmatrix}
 =
 Q_2
 \begin{pmatrix}
  19.6& 0 \\
  0&  0.161 \\
 \end{pmatrix}
 R_2
 \end{eqnarray*}
and
\begin{eqnarray*}
 \hat\Theta^{\rm (s)}
 &=&\begin{pmatrix}
 0.0783&  0.251&  - 0.716 \\
 0.753&  0.0591&  - 0.078 \\
 - 0.00341&  19.6&  0.0503 \\
\end{pmatrix}
=
  Q
  \begin{pmatrix}
  19.6& 0& 0 \\
   0&  0.815& 0 \\
   0& 0&  - 0.655 \\
  \end{pmatrix}
  R,
\end{eqnarray*}
where $Q_2$, $R_2$, $Q$ and $R$ are the matrices in (\ref{eq:asteroids-data-SVD}).
When we approximate the normalizing constant by the polynomial of degree $20$ 
as in the case of the comets' data and
use the output as a starting point of the holonomic gradient descent,
the HGD immediately finds a better likelihood value and then we 
reject the approximate solution by the series expansion method
of this degree.
We increase the degree of the approximation until the HGD does not reject
the output.
In the outputs above, we approximate the normalizing constant
by our series expansion truncated at the degree $40$.

We next compute the MLE by the HGD
with solving numerically the associated dynamical system
along gradient directions.
We add a superscript (h) as $\hat\Theta^{\rm (h)}$
for values computed by the holonomic gradient descent.
We use the output of the series expansion method as a starting point
of the HGD.
For the comets' data the MLE's
of the Fisher distribution on $\stiefel$ 
and on $SO(3)$ by the HGD
respectively agree with the MLE's by the series expansion method.

For the asteroid data, 
the MLE's are
$
 \hat\Theta^{(h)}_{1:2}=
 \hat\Theta^{(s)}_{1:2}$
and
$$
 \hat\Theta^{(h)} = 
 \begin{pmatrix}
 0.0779&  0.251&  - 0.733 \\
 0.769&  0.0591&  - 0.0776 \\
 - 0.00346&  19.6&  0.0505 \\
\end{pmatrix}
= Q
 \begin{pmatrix}
 19.6& 0& 0 \\
0&  0.831& 0 \\
0& 0&  - 0.671 \\
\end{pmatrix}
R.
$$
This output improves the approximate MLE by the series expansion method.

The AIC (Akaike Information Criterion) values are given in Table~\ref{table:aic}.
For each data set, AIC
was minimized by the Fisher distribution on $\sog$.
The log-likelihood ratio test statistic (LLR) of $\stiefel$ against $\sog$
and its $p$-value based on the $\chi^2$-approximation with 3 degrees of freedom are also shown.


\begin{table}[htbp]
 \caption{AIC of each data and each model, and the result of the likelihood ratio test of $\stiefel$ against $\sog$.}
 \label{table:aic}
 \begin{center}
  \begin{tabular}{c|ccc|ccc}
   & \multicolumn{3}{|c|}{comets}& \multicolumn{3}{|c}{asteroids} \\
   & Uniform& $\stiefel$& $\sog$& Uniform& $\stiefel$& $\sog$ \\
   \hline
   AIC& $0$& $-207.0$& $-219.7$& $0$& $-34764.6$& $-34769.2$ \\
   \hline
   LLR& & & $18.7$& & & $10.6$ \\
   $p$-value& & & $0.0003$& & & $0.014$
  \end{tabular}
 \end{center}
\end{table}



\newpage

\appendix
\section{Partial differential equation  for ${}_0F_1(p/2,Y)$}
\label{app:pde-hypergeometric}

If $\Theta=\mathrm{diag}(\theta_{ii})$ is diagonal, then
$y_i=\theta_{ii}^2/4$.
By change of variables from (\ref{eq:hyper-PDE}) we have
\begin{align}
 &y_i\Rd_i^2
 +\left\{\frac{p}{2}-\frac{r-1}{2}+\frac{1}{2}\sum_{j=1,j\neq i}^r\frac{y_i}{y_i-y_j}\right\}\Rd_i
 -\frac{1}{2}\sum_{j=1,j\neq i}^{r}\frac{y_j}{y_i-y_j}\Rd_j - 1
 \nonumber
 \\
 &\ =\ \Rd_{ii}^2+(p-r)\frac{1}{\theta_{ii}}\Rd_{ii}
 +\sum_{j\neq i}\frac{1}{\theta_{ii}^2-\theta_{jj}^2}\left\{\theta_{ii}\Rd_{ii}-\theta_{jj}\Rd_{jj}\right\}-1.
 \label{eqn:Muirhead}
\end{align}
We now show that the numerator of (\ref{eqn:Muirhead}) belongs $I_{\rm diag}$ for
small dimensions.

For $p=r=2$ (i.e.\ for $O(2)$)
by Macaulay2 we checked that the above $I$ is holonomic.
Also by asir (nk\_restriction), a set of generators of $I_{\rm diag}$ is given
as 
\begin{align*}
 & h_1 = (-\theta_{22}^2+\theta_{11}^2)\Rd_{11}^4 + 6\theta_{11}\Rd_{11}^3 + (2\theta_{22}^2-2\theta_{11}^2+6)\Rd_{11}^2
 -6\theta_{11}\Rd_{11}-\theta_{22}^2+\theta_{11}^2-3,
 \\
 & h_2 = (\theta_{22}^2-\theta_{11}^2)\Rd_{11}^2-\theta_{11}\Rd_{11}+\theta_{22}\Rd_{22}-\theta_{22}^2+\theta_{11}^2,
 \\
 & h_3 = \theta_{22}\Rd_{11}^4+\theta_{11}\Rd_{22}\Rd_{11}^3+(3\Rd_{22}-\theta_{22})\Rd_{11}^2-\theta_{11}\Rd_{22}\Rd_{11}-2\Rd_{22}
,
 \\
 & h_4 = \theta_{11}\theta_{22}\Rd_{11}^3+(\theta_{11}^2\Rd_{22}-\theta_{22})\Rd_{11}^2+(-\theta_{11}^2-1)\Rd_{22}+\theta_{22},
 \\
 & h_5 = -\Rd_{11}^2+\Rd_{22}^2.
\end{align*}
Looking at $h_2$ and $h_5$ we have
\begin{align*}
 &h_2\ =\ (\theta_{22}^2-\theta_{11}^2)
 \left\{\Rd_{11}^2+\frac{\theta_{11}\Rd_{11}-\theta_{22}\Rd_{22}}{\theta_{11}^2-\theta_{22}^2}
 -1\right\},
 \\
 &\frac{h_2}{\theta_{22}^2-\theta_{11}^2}+h_5\ =\ 
 \left\{\Rd_{22}^2+\frac{\theta_{22}\Rd_{22}-\theta_{11}\Rd_{11}}{\theta_{22}^2-\theta_{11}^2}
 -1\right\}.
\end{align*}
These are the same as  (\ref{eqn:Muirhead}) for $p=r=2$.

For the case of $V_2(\mbR^3)$ ($p=3$, $r=2$)
by Macaulay2 we have checked that $I$ is holonomic
By asir (nk\_restriction) $I_{\rm diag}$ has the set of generators:
\begin{align*}
 & h_1 = -\theta_{11}\Rd_{22}\Rd_{11}^2 + (-\theta_{22}\Rd_{22}^2-3\Rd_{22}+\theta_{22})\Rd_{11} + \theta_{11}\Rd_{22},
 \\
 & h_2 = \theta_{11}\theta_{22}\Rd_{11}^2+\theta_{22}\Rd_{11}-\theta_{11}\theta_{22}\Rd_{22}^2-\theta_{11}\Rd_{22},
 \\
 & h_3 = \theta_{11}^2\Rd_{11}^2+2\theta_{11}\Rd_{11}-\theta_{22}^2\Rd_{22}^2-2\theta_{22}\Rd_{22}+\theta_{22}^2-\theta_{11}^2,
 \\
 & h_4 = -\theta_{11}\Rd_{11}^2+(\theta_{22}^2\Rd_{22}^2+2\theta_{22}\Rd_{22}-\theta_{22}^2-1)\Rd_{11}
 +\theta_{11}\theta_{22}\Rd_{22}^3+2\theta_{11}\Rd_{22}^2-\theta_{11}\theta_{22}\Rd_{22},
 \\
 & h_5 = (-\theta_{11}\theta_{22}\Rd_{22}^2-\theta_{11}\Rd_{22}+\theta_{11}\theta_{22})\Rd_{11}-\theta_{22}^2\Rd_{22}^3
 -4\theta_{22}\Rd_{22}^2+(\theta_{22}^2-2)\Rd_{22}+2\theta_{22},
 \\
 & h_6 = -\theta_{11}\theta_{22}\Rd_{11}+(\theta_{22}^3-\theta_{11}^2\theta_{22})\Rd_{22}^2
 +(2\theta_{22}^2-\theta_{11}^2)\Rd_{22}-\theta_{22}^3+\theta_{11}^2\theta_{22}.
\end{align*}
Looking at $h_6$
\begin{align*}
 h_6&\ =\ -\theta_{11}\theta_{22}\Rd_{11}+(\theta_{22}^3-\theta_{11}^2\theta_{22})\Rd_{22}^2
 +(2\theta_{22}^2-\theta_{11}^2)\Rd_{22}-\theta_{22}^3+\theta_{11}^2\theta_{22}
 \\
 &\ =\ 
 (\theta_{22}^2-\theta_{11}^2)\theta_{22}
 \left\{
 \Rd_{22}^2 + \frac{2\theta_{22}^2-\theta_{11}^2}{(\theta_{22}^2-\theta_{11}^2)\theta_{22}}\Rd_{22}
 -\frac{\theta_{11}}{\theta_{22}^2-\theta_{11}^2}\Rd_{11}
 -1
 \right\}
 \\
 &\ =\ 
 (\theta_{22}^2-\theta_{11}^2)\theta_{22}
 \left\{
 \Rd_{22}^2 + \frac{1}{\theta_{22}}\Rd_{22}
 + \frac{\theta_{22}\Rd_{22}-\theta_{11}\Rd_{11}}{\theta_{22}^2-\theta_{11}^2}
 -1
 \right\}
\end{align*}
we see that it coincides with  the case of 
$p=3, r=2, i=2$ in (\ref{eqn:Muirhead}).

The case of $p=r=3$ is too big for the current implementation of the $D$-module
theoretic restriction algorithm in the nk\_restriction package.
Theorem \ref{th:extraeq}, which gives elements in $I_{\rm diag}$,
is shown by a different method.

\section{Asymptotic evaluation of $E(k,l,m)$}
\label{app:eklm}

We derive an asymptotic form of $E(k,l,m)$ when
$k,l,m$ simultaneously go to infinity.
Let $k=n\alpha$, $l=n\beta$ and $m=n\gamma$
where $\alpha$, $\beta$ and $\gamma$ are fixed positive numbers.
We use Laplace's method to show
\begin{equation}
 E(k,l,m)
  \sim \sqrt{\frac{2}{\pi}} ((k+l)(l+m)(k+m))^{-1/2}
  \label{E_expr_asymp}
\end{equation}
as $n\to\infty$.
The integrand $x_{11}^kx_{22}^lx_{33}^m$ of $E(k,l,m)$
is maximized at four points $(x_{11},x_{22},x_{33})=(1,1,1)$, $(-1,-1,1)$, $(-1,1,-1)$ and $(1,-1,-1)$
as long as $k,l,m$ are all even or all odd.
By symmetry it is sufficient to consider the neighborhood of $\diag(1,1,1)$,
where $X$ is approximated by
\[
X=
\begin{pmatrix}
(1 - \epsilon_1^2 - \epsilon_2^2)^{1/2}& -\epsilon_1 &  -\epsilon_2  \\
\epsilon_1 &(1 - \epsilon_1^2 - \epsilon_3^2)^{1/2}&   -\epsilon_3 \\
\epsilon_2 &  \epsilon_3 & (1 - \epsilon_2^2 - \epsilon_3^2)^{1/2}
\end{pmatrix}
\]
with sufficiently small numbers $\epsilon_1,\epsilon_2,\epsilon_3$.
The density of $(\epsilon_1,\epsilon_2,\epsilon_3)$ with respect to the Lebesgue measure
$d\epsilon_1d\epsilon_2d\epsilon_3$ is $1/{\rm Vol}(\sog)=1/(8\pi^2)$.
Hence we obtain
\begin{align*}
E(k,l,m)
& \sim 4\int (1 - \epsilon_1^2 - \epsilon_2^2)^{k/2}
 (1 - \epsilon_1^2-\epsilon_3^2)^{l/2}
 (1 - \epsilon_2^2-\epsilon_3^2)^{m/2}
 \frac{1}{8\pi^2}d\epsilon_1d\epsilon_2d\epsilon_3
 \\
& \sim 4\int {\rm e}^{-(k+l)\epsilon_1^2/2-(k+m)\epsilon_2^2/2-(l+m)\epsilon_3^2/2}
 \frac{1}{8\pi^2}d\epsilon_1d\epsilon_2d\epsilon_3
 \\
& = \sqrt{\frac{2}{\pi}}((k+l)(l+m)(k+m))^{-1/2}.
\end{align*}
We have checked that the right-hand side gives a good approximation to the exact value
of $E(k,l,m)$ for  $k+l+m \ge 100$.

The same argument shows that for $\sogp$
\begin{align*}
 E\left[\prod_{i=1}^p x_{ii}^{k_i}\right]
 \sim \frac{p(p-1)}{2}\frac{1}{{\rm Vol}(\sogp)}
 \left(\prod_{i<j}(k_i+k_j)\right)^{-1/2}
\end{align*}
as $n\to\infty$ when $k_i=n\alpha_i$, $\alpha_i>0$,  and $k_i$'s are 
are all even or all odd.



\bibliographystyle{plainnat}
\bibliography{rotation}

\end{document}